%% file: main.tex
\documentclass[11pt]{article}

\usepackage[
    persons={Aviad, Nima, Ruiquan},
    bibliosources={refs.bib},
    ifclass={{acmart,sigconf,sigconf}},
]{pomegranate}

\DeclareTheorem{question}
\DeclareDelimiter{\H}[\mathcal{H}]{\lparen}{\rparen}
\DeclareDelimiter{\tv}[\mathnormal{d}_{\operatorname{TV}}]{\lparen}{\rparen}
\DeclareDelimiter{\kl}[\mathcal{D}_{\operatorname{KL}}]{\lparen}{\rparen}
\DeclareDelimiter{\codim}[\operatorname{co\text{-}dim}]{\lparen}{\rparen}
\DeclareOperator{\ALG}

\title{Parallel Sampling via Counting}
\date{}

\Tag<anon>{
    \author{}
}
\Tag{
    \author{Nima Anari}
    \author{Ruiquan Gao}
    \author{Aviad Rubinstein}
    \affil{Stanford University, \url{{anari,ruiquan,aviad}@stanford.edu}}
}
\Tag<sigconf>{
	\author{Nima Anari}
	\email{anari@cs.stanford.edu}
	\orcid{0000-0002-4394-3530}
	\affiliation{
	  \institution{Stanford University}
	  \city{Stanford}
	  \state{California}
	  \country{USA}
	}
	
	\author{Ruiquan Gao}
	\email{ruiquan@cs.stanford.edu}
	\orcid{0009-0006-9837-8598}
	\affiliation{
	  \institution{Stanford University}
	  \city{Stanford}
	  \state{California}
	  \country{USA}
	}
	
	\author{Aviad Rubinstein}
	\email{aviad@cs.stanford.edu}
	\orcid{0000-0002-6900-8612}
	\affiliation{
	  \institution{Stanford University}
	  \city{Stanford}
	  \state{California}
	  \country{USA}
	}
}

\Tag<sigconf>{
	\input{concepts}
	
\setcopyright{acmlicensed}
\acmDOI{10.1145/3618260.3649744}
\acmYear{2024}
\copyrightyear{2024}
\acmSubmissionID{stoc24main-p931-p}
\acmISBN{979-8-4007-0383-6/24/06}
\acmConference[STOC '24]{Proceedings of the 56th Annual ACM Symposium on Theory of Computing}{June 24--28, 2024}{Vancouver, BC, Canada}
\acmBooktitle{Proceedings of the 56th Annual ACM Symposium on Theory of Computing (STOC '24), June 24--28, 2024, Vancouver, BC, Canada}
\received{13-NOV-2023}
\received[accepted]{2024-02-11}
}

\Tag<sigconf>{
	\DeclareFieldFormat{doi}{%
		\ifhyperref
	    	{\href{https://doi.org/##1}{\nolinkurl{https://doi.org/##1}}}
		    {\nolinkurl{https://doi.org/##1}}}
}

\begin{document}
	\Tag<sigconf>{
		\begin{abstract}
			\input{abstract}
		\end{abstract}
		\keywords{parallel sampling, counting, conditional marginals, autoregressive models}
		\maketitle
	}
	\Tag{
    	\maketitle
    	\begin{abstract}
    		\input{abstract}
    	\end{abstract}
		\clearpage
	}

    \input{intro}
    \input{prelim}
    \input{algorithm}
    \input{application}
    \input{hardness}

	\Tag{\input{algo-lb}}

	\Tag<sigconf>{\input{acmacks}}
    \Tag{\PrintBibliography}

    \appendix
    
    \Tag{\input{app-proof}}
    
    \Tag<sigconf>{\PrintBibliography}
\end{document}

%% file: concepts.tex
\ccsdesc[300]{Theory of computation~Generating random combinatorial structures}
\ccsdesc[500]{Theory of computation~Parallel algorithms}
\ccsdesc[300]{Mathematics of computing~Probabilistic algorithms}

%% file: abstract.tex
We show how to use parallelization to speed up sampling from an arbitrary distribution $\mu$ on a product space $[q]^n$, given oracle access to counting queries: $\mathbb{P}_{X\sim \mu}[X_S=\sigma_S]$ for any $S\subseteq [n]$ and $\sigma_S \in [q]^S$. Our algorithm takes $O({n^{2/3}\cdot \operatorname{polylog}(n,q)})$ parallel time, to the best of our knowledge, the first sublinear in $n$ runtime for arbitrary distributions. Our results have implications for sampling in autoregressive models. Our algorithm directly works with an equivalent oracle that answers conditional marginal queries $\mathbb{P}_{X\sim \mu}[X_i=\sigma_i\;\vert\; X_S=\sigma_S]$, whose role is played by a trained neural network in autoregressive models. This suggests a roughly $n^{1/3}$-factor speedup is possible for sampling in any-order autoregressive models. We complement our positive result by showing a lower bound of $\widetilde{\Omega}(n^{1/3})$ for the runtime of any parallel sampling algorithm making at most $\operatorname{poly}(n)$ queries to the counting oracle, even for $q=2$.

%% file: intro.tex
\section{Introduction}\label{sec:intro}

The seminal work of \textcite{JVV86} established an algorithmic equivalence between the tasks of \emph{approximate sampling} and \emph{approximate counting}, for the ubiquitous class of self-reducible problems. This key equivalence is at the heart of the Monte Carlo Markov Chain approach to approximate counting \cite{SVV09}, which has enabled breakthroughs like approximating the permanent \cite{JSV04} or the volume of convex sets \cite{DFK91}. In this paper, we focus on one side of this equivalence, sampling via counting.

Self-reducibility, in its most widely applied form, concerns distributions $\mu$ on a product space $[q]^n$ and their pinnings: conditional distributions obtained by selecting a subset $S\subseteq[n]$ and partial configuration $\sigma_S\in [q]^S$ and conditioning $X\sim \mu$ to have coordinates in $S$ pinned to $\sigma_S$: $X_S=\sigma_S$. In this setting, sampling means producing a random $X$ distributed according to a specified pinning of $\mu$. Counting, on the other hand, refers to computing the partition functions of pinnings: $\P_{X\sim \mu}{X_S=\sigma_S}$.\footnote{In the literature, often $\mu$ is assumed to be an unnormalized measure. The partition function for unnormalized measures is simply $\mu\parens*{\set{X\in [q]^n\given X_S=\sigma_S}}$. Counting algorithms for unnormalized measures and normalized measures are easily reducible to each other, so w.l.o.g.\ we assume $\mu$ is normalized.} Sampling via counting is in fact very easy to describe in this setting. Assuming access to a counting oracle, we can produce samples from $\mu$ via the following \emph{autoregressive generation process}:
\begin{Algorithm*}
	Initialize $\sigma\gets \emptyset$\;
	\For{$i=1,\dots,n$}{
		\For{$x\in [q]$}{
			$p_x\gets \P*_{X\sim \mu}{X_i=x\given X_{[i-1]}=\sigma_{[i-1]}}$\;
		}
		$\sigma_i\gets$random sample in $[q]$ distributed $\sim(p_1,\dots,p_q)$\;
	}
	\Return{$\sigma$}
\end{Algorithm*}
Note that we only need to use the counting oracle to compute the computationally equivalent \emph{conditional marginals}:
\[ \P*_{X\sim \mu}{X_i=x\given X_{[i-1]}=\sigma_{[i-1]}}=\Tag{\frac{\P*_{X\sim \mu}{X_{[i-1]}=\sigma_{[i-1]}, X_i=x}}{\P*_{X\sim \mu}{X_{[i-1]}=\sigma_{[i-1]}}}}\Tag<sigconf>{\frac{\P*{X_{[i-1]}=\sigma_{[i-1]}, X_i=x}}{\P*{X_{[i-1]}=\sigma_{[i-1]}}}}. \]

This process, despite its simplicity, is how the widely successful autoregressive models generate their output \cite[see, e.g.,][]{LM11,VKK16,VSPUJLAKP17,DCLT18,YDYCSL19,BOpenAI20}. State-of-the-art large language models, or even some competitive vision models, train a neural network to answer \emph{conditional marginal queries} and then use the aforementioned process to generate samples. In the context of language models, $[q]$ represents a token space, and $n$ is the length of generated text or context length, while in pixel-space vision models, $[q]$ is possible values for a pixel, and $n$ is the number of pixels in the image.

One downside of this simple sampling process is that it is extremely sequential. One has to generate coordinates $1,\dots,i-1$, to know which conditional marginals need to be queried in the $i$th iteration. So, it is natural to ask if there is a more \emph{parallelizable} sampling process. More precisely, suppose that an oracle\footnote{E.g., a neural network in learned autoregressive models.} can answer conditional marginal queries of the form $\P{X_i=x\given X_S=\sigma_S}$, and we can interact with this oracle in rounds, each time asking polynomially many queries simultaneously. We are interested in finding the \emph{adaptive complexity} of sampling:
\begin{question}What is the minimum number of rounds before we can produce a sample?\end{question}

At first glance, it might seem that $\simeq n$ is roughly the optimal number of rounds. Indeed, if we are restricted to asking queries where we always pin a prefix $X_{[i-1]}$ and ask for the conditional marginal of the next coordinate $X_i$, not much better is possible. Imagine the adversarially chosen distribution $\mu$ being a Dirac delta on a single randomly chosen $\sigma \in [q]^n$. One cannot ``guess'' more than $\widetilde{O}(1)$ coordinates of $\sigma$ at a time, and thus any query pinning more than $\widetilde{O}(1)$ new coordinates is useless. Thus it takes $\widetilde{\Omega}(n)$ rounds to find the hidden $\sigma$.

Perhaps surprisingly, we show that when pinning is allowed on \emph{any subset} of the coordinates, we can significantly improve over $n$. For details of the algorithm, see \cref{sec:alg}.
\begin{theorem}[Main]\label{thm:main}
	There is an algorithm that produces a random sample from any distribution $\mu$ on $[q]^n$ by interacting in rounds with an oracle that answers conditional marginal queries, with each query returning $\P_{X\sim \mu}{X_i=x \given X_S=\sigma_S}$ for all $x\in [q]$. The total number of queries is $O(n)$, and the expected number of rounds is \[O(n^{2/3}\cdot \min\{\log^{2/3}{n}\cdot \log{q},q^{1/3}\log^{1/3}{q}\}).\]
\end{theorem}
We note that, although we mostly care about parallelizing interactions with the oracle, our algorithm's internal computation can also be parallelized, and up to polylogarithmic factors, the runtime on a \Class{PRAM} would be the same as the bound in \cref{thm:main}. We also note that the guarantee on the expected number of rounds for our algorithm also holds with high probability at the cost of extra logarithmic factors, see \cref{thm:tail-bounds}.

\begin{remark}
	In autoregressive models, especially large language models, $q$ is usually very large, but \cref{thm:main} has a mild dependency on $q$, at most logarithmic. Since autoregressive models are often run on hardware already capable of massive amounts of parallelism, e.g., GPUs or TPUs, one can expect our algorithm to speed up generation time even in practice. We leave experimental evaluation to future works, but we also note two potential issues. First, while many autoregressive models, such as XLNet \cite{YDYCSL19} or generally any-order autoregressive models \cite{SSE22}, allow pinning of any subset, many others only allow pinning of prefixes; as noted before, no significant parallel speedup is possible for just prefix pinnings. Second, in practice, the oracle's role is played by a trained neural network, which clearly returns only approximate answers. While we can handle approximate oracles, \Tag<sigconf>{see the full version of this article,}\Tag{see \cref{sec:approximate-oracle},} the guarantees we need in theory might not hold in practice.
\end{remark}

One might wonder if the number of rounds can be further improved, perhaps by using a different algorithm. In a dream scenario, would a polylogarithmic number of rounds be feasible? We answer this question \emph{negatively}, by providing a lower bound of $\widetilde{\Omega}(n^{1/3})$ for any algorithm.
\begin{theorem}[Lower bound, informal]\label{thm:lowerbound-informal}
	Even for $q=2$, any algorithm sampling from arbitrary distributions on $[q]^n$ needs to interact with the conditional marginal oracle for at least $\widetilde{\Omega}(n^{1/3})$ rounds.
\end{theorem}
For the more formal statement of our lower bound, see \cref{sec:hardness}. This shows that the optimal number of rounds, while sublinear in $n$, must still be a polynomially large function of $n$, at least with no further assumption on the distribution $\mu$.

\subsection{Related Work}

Interest in parallel sampling started decades ago. As an early example, \textcite{MVV87}, having found an algorithm to generate perfect matchings in parallel, asked if a \emph{uniformly random} one can also be generated in parallel. \Textcite{Ten95} provided negative evidence for this. Recently, there has been a significantly increased interest in parallel sampling algorithms.

Markov chains, arguably the most successful sampling tool, are na\"ively sequential, but recent works have shown techniques for parallelizing some classes of Markov chains, including Glauber dynamics, under tractability conditions on the distribution $\mu$ \cite{FHY21,LY22,Lee23}. We note that even sequential implementations of Markov chains such as Glauber dynamics take exponential time on \emph{worst-case distributions} $\mu$, so it is natural that these works need further assumptions on $\mu$.

Most related to our work, parallel sampling was raised as an open question by \textcite{AHSS20} for several challenge distributions that admit parallel (\Class{NC}, i.e., polylogarithmic time on polynomially many machines) counting algorithms. These include the distributions of uniformly random arborescences, directed Eulerian tours, planar perfect matchings, and determinantal point processes. They showed polylogarithmic sampling is possible for one of these challenges: sampling uniformly random arborescences. Later, \textcite{ABTV23} showed polynomial parallel speedups are possible for the class of \emph{entropically independent} distributions, which included all challenges except for \emph{planar perfect matchings}. Most recently, \textcite{AHLVXY23} achieved polylogarithmic sampling for all challenge distributions except for \emph{planar perfect matchings}, using the stronger ``weighted counting oracle.'' This stronger oracle returns marginals not just after pinnings, but under all ``exponential tilts,'' and interestingly, is what another class of generative AI models, namely \emph{diffusion models}, attempt to learn.

We note that all of these prior works use some tractability assumption about the distribution $\mu$. In fact, none of them are able to nontrivially speed up sampling of \emph{planar perfect matchings}, one of the original challenges. In contrast, in our work, the emphasis is on \emph{arbitrary distributions} $\mu$, as none of the tractability assumptions of prior work is likely to hold for example by distributions learned by autoregressive models. As an application of our results, we show how to nontrivially speed up parallel sampling of \emph{planar perfect matchings} in \cref{sec:applications}.

Recently, generative modeling in AI has produced amazing results. State-of-the-art models, depending on the domain or modality, are often autoregressive or diffusion-based. Given their huge importance in practice, significant attention has been paid to improving the sampling efficiency of these models, particularly via parallelism. For example, Picard iterations in diffusion models \cite{SBESA23} and speculative decoding in autoregressive models \cite{CBILSJ23,LKM23} have shown practical accelerations. There are many other techniques introduced in the literature, evaluated experimentally, by way of example ``prediction and forecasting'' \cite{WH20} and fixed-point iterations based on Jacobi and Gauss-Seidel equations \cite{SMLE21}. To the best of our knowledge, these works focus on real-world distributions and do not theoretically prove an unconditional asymptotic parallel speedup. Interestingly, some of these practical parallelization techniques, for example, speculative decoding, share similarities with our sampling algorithm, \cref{alg:sample-on-hypergrid2}. In speculative decoding, a \emph{draft model}, a much faster but less accurate model, is used to generate guesses sequentially for future tokens and these guesses are ``verified'' using a larger but more accurate model in parallel. Our algorithm is also based on a guessing and verification paradigm but differs from speculative decoding because we cannot afford to sequentially run a draft model. We emphasize that our work is focused on theoretical guarantees, and works with a single oracle, not tiers of oracles with cost/accuracy tradeoffs.

Finally, adaptive complexity has been studied for many other computational problems, for example, submodular maximization \cite{BRS19,LLV20} and minimization \cite{BS20,CCK21,CGJS22}. Most notably, the parallel complexity of search via a decision oracle was studied in the seminal work of \textcite{KUW88}, who showed, similarly to our results, that a polynomial speedup, and no better than a polynomial speedup, was possible. 

\subsection{Techniques}

Our algorithm works by modifying the autoregressive sampling process in two ways. First, we choose the order of coordinates according to a uniformly random permutation. Second, to break sequentiality, we generate ``guesses'' of future coordinates, by computing the marginal of each $X_i$ conditioned on the current pinning, \emph{in parallel}, and sampling from these marginal distributions independently for each $i$. While these independent samples clearly ignore dependencies between coordinates, we can in a second stage \emph{verify in parallel} that each $X_i$ would have been the sample produced if we had continued sequentially. We advance up to the point where our guesses successfully pass verification and then iterate.

The key idea behind our analysis is that as we pin more and more random coordinates, the dependencies between the remaining coordinates weaken in an \emph{average sense}. This intuitive idea is formalized by the so-called \text{pinning lemmas} \cite{RT12,Mon08} which we use in our analysis, see \cref{sec:pinning}. Weakened dependencies intuitively mean that our guesses are not likely to deviate from what sequential sampling would have produced. This is formally proved in \cref{sec:alg}.

Finally, a tool we use from existing literature on parallel sampling is a universal coupler \cite{LY22}. This is used to ensure consistency between the guessing and verification stages. In both of these phases, for each $X_i$, we would like to sample from a marginal distribution. Universal couplers ensure that when the marginal distributions are ``close'', the samples are likely to be exactly equal. We extend the analysis of universal coupling to multiple distributions, as needed by our work, see \cref{sec:universal-coupling}.

To prove our lower bound, we construct a challenge distribution that is a uniform distribution on an affine subspace of $\F_2^n=\set{0, 1}^n$, and we show that it is hard to even output \emph{anything} in its support in fewer than $\widetilde{\Omega}(n^{1/3})$ rounds. We group the coordinates in $[n]$ into roughly $\widetilde{\Omega}(n^{1/3})$ randomly chosen buckets and put varying numbers of affine constraints on each bucket. We prove that with high probability the buckets can only be discovered one at a time, from the most constrained bucket to the least. This is because queries pinning too many coordinates will not be useful at all, as they will violate the constraints of the most constrained bucket. On the other hand, if the number of pinnings is just right for the most-constrained undiscovered bucket, no information is gained about less-constrained buckets; with high probability all of the marginals in the less-constrained buckets remain uniform.

\subsection{Organization}

In \cref{sec:prelims} we discuss and further develop two of the main tools we use for parallelization: pinning lemmas and universal coupling. In \cref{sec:alg}, we describe our parallel sampling algorithm and prove our main result \cref{thm:main}. In \cref{sec:applications}, we provide an application of \cref{thm:main} to the problem of sampling planar perfect matchings. In \cref{sec:hardness}, we prove a lower bound against all algorithms, i.e., \cref{thm:lowerbound-informal}. \Tag{In \cref{sec:tightness}, we prove that our analysis of the algorithm presented in \cref{sec:alg} is tight, and $n^{2/3}$ cannot be improved for this particular algorithm.}

\Tag{
	\subsection*{Acknowledgments}
	\input{acks}
}

%% file: acks.tex
Nima Anari was supported by NSF CAREER Award CCF-2045354. Ruiquan Gao was supported by NSF CCF-1954927, and a Stanford Graduate Fellowship. Aviad Rubinstein was supported by NSF CCF-1954927, and a David and Lucile Packard Fellowship.

%% file: prelim.tex
\section{Preliminaries}\label{sec:prelims}

We use $[n]$ to denote the set $\braces{1,2,\cdots, n}$. 
For any vector $x$ and set $S$, we use $x_S$ to denote $x$ restricted to $S$. We use $\pm x$ to indicate the interval $[-x,x]$. We use $\mathcal{S}_n$ to denote the set of permutations on $n$ elements. For any two sets $A,B$, we use $A\times B$ to denote the Cartesian product of $A$ and $B$, i.e., $A\times B = \braces{(a,b) \given a\in A, b\in B}$.

For a distribution $\mu$, we use $x\sim \mu$ to denote that $x$ is sampled from $\mu$. Similarly, for a set $S$, we use $x\sim S$ to indicate that $x$ is sampled uniformly at random from $S$.

\subsection{Pinning Lemmas}\label{sec:pinning}

The pinning lemma formalizes an intuition that randomly pinning coordinates of an arbitrary distribution should in an average sense lower the correlation between remaining coordinates. Intuitively, this should make parallel sampling easier; for example, if all coordinates become fully independent, one can in parallel sample from the marginals. We do not directly use the classical statement of the pinning lemma, mentioned below for comparison, but rather prove a statement in the same vein and using the same proof strategy.

\begin{lemma}[pinning lemma \cite{RT12,Mon08}]
Let $X_1, X_2,\dots, X_n$ be random variables, each supported on $\set{0,1}$.
For any $\l\in [n]$, there exists a set $S$ such that $\card{S}\leq \l$ and
\[
    \E*_{X_S}{
    \E*_{
        u,v\sim \binom{[n]}{2}
    }{
        \operatorname{Cov}
        \parens*{
            X_u,X_v
            \given
            X_S
        }^2
    }}
    \leq 
    \frac{O(1)}{\l}.
\]
\end{lemma}

We now define and state well-known statements about entropy, building up to state and prove our new variant of the pinning lemma.

\begin{definition}[entropy]
Let $X,Y$ be random variables on $[q]$. 
The entropy of random variable $X$ is defined to be
\begin{align*}
    \H{X} 
    = 
    -
    \sum_{i \in [q]} 
    \P *{X=i} \cdot \log \P *{X=i}.
\end{align*}
The conditional entropy of $X$ conditioned on $Y$ is defined to be
\begin{align*}
    \H{X \given Y}
    =
    \sum_{i\in [q]} \H{X \given Y=i} \cdot \P{Y=i}.
\end{align*}
\end{definition}

\begin{definition}[KL divergence]
    For a pair of distributions $\nu, \mu$, we let
    \[\kl{\nu\river \mu}=\E*_{x\sim \nu}{\log\frac{\nu(x)}{\mu(x)}}.\]
    Abusing notation, we extend the definition to random variables. If $X\sim \nu, Y\sim \mu$, we use $\kl{X\river Y}$ to denote $\kl{\nu\river \mu}$.
\end{definition}

\begin{lemma}
    \label{lem:kl-vs-entropy}
    For any two random variables $X,Y$, 
    \begin{align*}
        \E*_{Y}{\kl*{(X\given Y) \river X}} = \H{X} - \H{X\given Y}.
    \end{align*}
\end{lemma}

\begin{lemma}[Pinsker's inequality]
    \label{lem:tv-vs-kl}
    For any two random variables $X,Y$, 
    \begin{align*}
        \tv{X,Y} \leq \sqrt{\frac{1}{2}\kl{X \river Y}}
    \end{align*}
\end{lemma}

We prove and use the following variant of the pinning lemma. 
\begin{lemma}
    \label{cor:pinning-on-sqr-prob-diff}
    For any integer $\theta>0$ and any collection of random variables $X=(X_1,\dots,X_n)$ with support $[q]$,\Tag<sigconf>{ if $\sigma\sim \mathcal{S}_n$,}
    \Tag{
	    \[
	        \E*_{X, \sigma\sim \mathcal{S}_n}{
	        	\sum_{i=\theta}^{n} 
	            \tv*{
	                X_{\sigma(i)} \given \set*{X_{\sigma(j)}}_{j\in [i-\theta]}
	                ,
	                X_{\sigma(i)} \given \set*{X_{\sigma(j)}}_{j\in [i-1]}
	            }^2
	        }\\
	        \leq 
	        \frac{(\theta-1)\log q}{2}.
	    \]
	}%
    \Tag<sigconf>{\begin{multline*}
        \E*{
        	\sum_{i=\theta}^{n} 
            \tv*{
                X_{\sigma(i)} \given \set*{X_{\sigma(j)}}_{j\in [i-\theta]}
                ,
                X_{\sigma(i)} \given \set*{X_{\sigma(j)}}_{j\in [i-1]}
            }^2
        }\\
        \leq 
        \frac{(\theta-1)\log q}{2}.
    \end{multline*}}%
\end{lemma}
\begin{proof}
    For any permutation $\sigma \in \mathcal{S}_n$ and any $i\in \set{\theta, \theta+1, \cdots, n}$, 
    \begin{multline*}
        \E*_{X}{\tv *{X_{\sigma(i)}\given \braces{X_{\sigma(j)}}_{j\in [i-\theta]}, X_{\sigma(i)} \given \braces{X_{\sigma(j)}}_{j\in [i-1]}}^2}
        \leq\\
        \frac{\E*_{X}{\kl *{X_{\sigma(i)}\given \braces{X_{\sigma(j)}}_{j\in [i-\theta]} \river X_{\sigma(i)} \given \braces{X_{\sigma(j)}}_{j\in [i-1]}}}}{2}
		\leq\\
        \frac{
            \H*{X_{\sigma(i)} \given \braces{X_{\sigma(j)}}_{j\in [i-\theta]}}
            -
            \H*{X_{\sigma(i)} \given \braces{X_{\sigma(j)}}_{j\in [i-1]}}
        }{2}
		=\\
        \sum_{k=i-\theta}^{i-2}\frac{{
            \H*{X_{\sigma(i)} \given \braces{X_{\sigma(j)}}_{j\in [k]}}
            -
            \H*{X_{\sigma(i)} \given \braces{X_{\sigma(j)}}_{j\in [k+1]}}
        }}{2}.
    \end{multline*}
    Summing over all possible $i$ and taking expectation over all permutations, we have
    \begin{multline*}
        \sum_{i=\theta}^{n} 
        \E*{
            \tv*{
                X_{\sigma(i)} \given \braces*{X_{\sigma(j)}}_{j\in [i-\theta]}
                ,
                X_{\sigma(i)} \given \braces*{X_{\sigma(j)}}_{j\in [i-1]}
            }^2
        }
        \leq\\
        \frac12 \cdot 
        \sum_{i=\theta}^{n} 
        \sum_{k=i-\theta}^{i-2}
            \E*_{\sigma\sim \mathcal{S}_n}{
                \H*{X_{\sigma(i)} \given \braces{X_{\sigma(j)}}_{j\in [k]}}
            }
            -
            \Tag<sigconf>{\\}
            \E*_{\sigma\sim \mathcal{S}_n}{
                \H*{X_{\sigma(i)} \given \braces{X_{\sigma(j)}}_{j\in [k+1]}}
            }
        \leq\\
        \frac12 \cdot
        \sum_{i=\theta}^{n} 
        \sum_{k=i-\theta}^{i-2}
            \E*_{\sigma\sim \mathcal{S}_n}{
                \H*{X_{\sigma(n)} \given \braces{X_{\sigma(j)}}_{j\in [k]}}
            }
            -
            \Tag<sigconf>{\\}
            \E*_{\sigma\sim \mathcal{S}_n}{
                \H*{X_{\sigma(n)} \given \braces{X_{\sigma(j)}}_{j\in [k+1]}}
            }
        \leq\\
        \frac{\theta-1}{2} \cdot
        \sum_{k=0}^{n-2} 
            \E*_{\sigma\sim \mathcal{S}_n}{
                \H*{X_{\sigma(n)} \given \braces{X_{\sigma(j)}}_{j\in [k]}}
            }
            -
            \Tag<sigconf>{\\}
            \E*_{\sigma\sim \mathcal{S}_n}{
                \H*{X_{\sigma(n)} \given \braces{X_{\sigma(j)}}_{j\in [k+1]}}
            }
        =\\
        \frac{\theta-1}{2} \cdot \E *_{\sigma\sim \mathcal{S}_n}{\H*{X_{\sigma(n)}} - \H*{X_{\sigma(n)} \given \braces{X_{\sigma(j)}}_{j\in [n-1]}}}
        \leq
        \Tag<sigconf>{\\}
        \frac{(\theta-1)\log q}{2}.
        \qedhere
    \end{multline*}
\end{proof}

\subsection{Universal Coupling}\label{sec:universal-coupling}

For any integer $q>0$, let $\Delta_q$ be the probability simplex on $[q]$, i.e., $\Delta_q = \set{\mu\in [0,1]^q\given \sum_{i=1}^q \mu(i)=1}$. In our main algorithm, we use a ``universal coupler'' as a subroutine. Informally, this is an algorithm that maps a distribution $u\in \Delta_q$ and a random source $r$ to a sample from $u$, with the property that the output is unlikely to change if $u$ is perturbed slightly (while keeping $r$ fixed). Such an algorithm naturally induces a coupling between any two distributions $\mu, \mu'$. In any coupling, there must be at least $\tv{\mu, \mu'}$ chance that the samples for $\mu, \mu'$ are unequal; and this lower bound can be achieved if we design a tailor-made coupling knowing both $\mu$ and $\mu'$. Surprisingly, one can achieve the same bound within constant factors \emph{without knowing both distributions in advance}.

The existence of these robust universal couplers appears to have been discovered and rediscovered many times. The earliest works that we are aware of are the MinHash algorithm of \textcite{Bro97} for uniform distributions, and a rejection-sampling-based strategy of \textcite{KT02,Hol07} for general distributions. See \cite{Bav20} for more on the history and optimality of these strategies. Recently, in the context of parallel sampling algorithms, the work of \textcite{LY22} has rediscovered the same rejection-sampling-based algorithm; we borrowed the terminology of ``universal coupling'' from the latter work.

A universal coupler is defined as follows:
\begin{definition}[universal coupling,~\cite{LY22}]
    A deterministic function $f:\Delta_q\times [0,1]\to [q]$ is a universal coupling on $[q]$ if, when $r\in [0,1]$ is chosen uniformly at random, for any distribution $\mu\in \Delta_q$ and $x\in [q]$,
    \[
        \P_{r\sim [0,1]}{f(\mu,r)=x} = \mu(x).
    \]
\end{definition}
Note that instead of $r\sim [0,1]$, one can use other sources of randomness with infinite entropy, such as an infinite sequence of random bits, etc. Since it is easy to translate between these sources, we pick the notationally most convenient form of random source when describing each universal coupler.

The main characteristic we would like for universal couplers is that on ``close'' distributions $\mu, \mu'$, the chance that $f(\mu, r)\neq f(\mu', r)$ is small. A lower bound on this chance is $\P_r{f(\mu, r)\neq f(\mu', r)}\geq \tv{\mu, \mu'}$. Surprisingly, this can be matched up to a factor of $2$; in fact the optimal $f$ has been shown \cite{KT02,Hol07,Bav20,LY22} to satisfy
\[ \P_r{f(\mu, r)\neq f(\mu', r)}\leq \frac{2\tv{\mu, \mu'}}{1+\tv{\mu, \mu'}}\leq 2\tv{\mu, \mu'}.\]

In our algorithm, we need a slightly stronger guarantee that holds for not just two, but an arbitrary number of distributions.
\begin{definition}[robust universal coupler]
	\label{def:robust-universal-coupler}
	We call a universal coupler $f$ robust if for any number of distributions $\mu_1,\dots, \mu_m$ it satisfies
	\[ \P*_r{\exists i,j\in [m]: f(\mu_i, r)\neq f(\mu_j, r)}\leq \frac{\sum_{x\in [q]} \parens{\max_{i\in [m]}{\mu_i(x)} - \min_{i\in [m]}{\mu_i(x)}}}{\sum_{x\in [q]} \max_{i\in [m]}{\mu_i(x)}} \]
\end{definition}
Note that for $m=2$, the numerator on the r.h.s.\ becomes $2\tv{\mu_1,\mu_2}$, and the denominator becomes $1+\tv{\mu_1, \mu_2}$. Thus, this matches the same optimal bound derived by prior works. We show that the rejection-sampling-based algorithm used in prior works, which we call the \textsc{MinCoupler}, satisfies this more general robustness guarantee. Additionally, we show that another widely used algorithm called the ``Gumbel trick'' also satisfies the same robustness guarantee.

\paragraph{Universal coupler of \textcite{KT02,Hol07}, rediscovered by \textcite{LY22}.} Interpret the uniformly random $r$ as a sequence of i.i.d.\ pairs $(x_1,p_1),(x_2,p_2), \cdots \in [q]\times [0,1]$, distributed uniformly at random.
Given the distribution $\mu$, the algorithm $f$ picks the smallest index $i\geq 1$ such that $p_i\leq \mu(x_i)$ and outputs $x_i$.
See \cref{alg:uni-coupling}.

\begin{Algorithm}
\caption{Universal coupler $\textsc{MinCoupler}$ \cite{KT02,Hol07,LY22}}
\label{alg:uni-coupling}
\KwIn{Distribution $\mu \in \Delta_q$, randomness $r\in [0,1]$}
\KwOut{A sample from $\mu$}
Interpret $r$ as uniform i.i.d.\ pairs $(x_1,p_1), (x_2,p_2), \cdots \in [q]\times [0,1]$, e.g., by
\begin{itemize}
    \item using bits of $r$ at odd indices for $x_1, x_2, \cdots $, and
    \item using bits of $r$ at even indices for $p_1, p_2, \cdots $ in a zig-zag order.
\end{itemize}
$i^*\gets \min\set{i\given p_i\leq \mu(x_i)}$ \;
\Return{$x_{i^*}$}
\end{Algorithm}

\paragraph{Correctness and efficiency.} \Cref{alg:uni-coupling} is a universal coupler and can be implemented with high probability by choosing only a sequence of $L=O(q\log n)$ pairs $(x_1,p_1),\cdots, (x_{L}, p_{L})$ \cite[see, e.g.,][]{LY22}. 
\begin{lemma}[{\cite[see, e.g., Lemma~4.3,][]{LY22}}]
    Suppose \textsc{MinCoupler} is constructed in \cref{alg:uni-coupling} and $i^*$ is the smallest index chosen by $\textsc{MinCoupler}$.
    Then, for any distribution $\mu\in \Delta_q$,
    \begin{enumerate}
        \item \textsc{MinCoupler} is a universal coupler: $\forall x\in [q]$, \[\P{\textsc{MinCoupler}(\mu,r)=x}=\mu(x),\]
        \item $i^*$ follows the geometric distribution with success probability $1/q$.
    \end{enumerate}
\end{lemma}

\paragraph{Performance.} 
For any two distributions $\mu,\upsilon\in \Delta_q$, the samples produced by \cref{alg:uni-coupling} for the two distributions (using a shared random number $r$) are different with probability at most $\frac{2\tv{\mu,\upsilon}}{1+\tv{\mu,\upsilon}}$, which is also tight in the worst case \cite[see][]{KT02,Hol07,Bav20,LY22}. 

We generalize the performance guarantee to a multi-distribution setting. 
\begin{lemma}[robustness of \textsc{MinCoupler}]
    \label{lem:uni-coupler-gurantee}
    Consider any $m,q>0$.
    Suppose 
    $r\in [0,1]$ is uniformly random.
    For any distributions $\mu_1,\cdots, \mu_m\in \Delta_q$, the probability that there exist $i,j\in [m]$ such that $ \textsc{MinCoupler}(\mu_i,r) \neq \textsc{MinCoupler}(\mu_j,r)$ is at most 
    \begin{align*}
        \frac{\sum_{x\in [q]} \parens{\max_{i\in [m]}{\mu_i(x)} - \min_{i\in [m]}{\mu_i(x)}}}{\sum_{x\in [q]} \max_{i\in [m]}{\mu_i(x)}}.
    \end{align*}
\end{lemma}
\begin{proof}
    For each $j\in [m]$, let $i^*_j$ be the smallest index the universal coupler chooses for $\mu_j$, i.e.,
    \begin{align*}
        i^*_j := \min{\braces*{i\given p_i\leq \mu_j(x_i)}}.
    \end{align*}
    If all $i^*_j$ are identical, the coupler's outputs on $\mu_1,\dots,\mu_m$ are the same. Therefore, 
    \begin{multline*}
        \P*_{r}{\exists i,j\in [m] \textsc{MinCoupler}(\mu_i,r) \neq \textsc{MinCoupler}(\mu_j,r)} 
        \leq \\
        1 - \P*{i^*_1=i^*_2=\cdots=i^*_m}
        = \P*{\min_{j\in [m]}{i^*_j} \neq \max_{j\in [m]}{i^*_j}}.
    \end{multline*}
    Observe that 
    \Tag{
    \[
    \min_{j\in [m]} i^*_j = \min{\braces*{i \given p_i\leq \max_{j\in [m]}{\mu_j(x_i)}}}
    \quad \text{and}\quad\Tag<sigconf>{\\}
    \max_{j\in [m]} i^*_j = \min{\braces*{i \given p_i\leq \min_{j\in [m]}{\mu_j(x_i)}}}.
    \]
    }%
    \Tag<sigconf>{
    \begin{multline*} 
    \min_{j\in [m]} i^*_j = \min{\braces*{i \given p_i\leq \max_{j\in [m]}{\mu_j(x_i)}}}
    \quad \text{and}\quad\Tag<sigconf>{\\}
    \max_{j\in [m]} i^*_j = \min{\braces*{i \given p_i\leq \min_{j\in [m]}{\mu_j(x_i)}}}.
    \end{multline*}
    }%
    We can thus upper bound 
    \begin{multline*}
        \P*{\max_{j\in [m]} i^*_j\neq \min_{j\in [m]} i^*_j}
        =
        \sum_{i\geq 1} \P*{\min_{j\in [m]}{i^*_j} = i}\cdot \Tag<sigconf>{\\} \P*{p_i> \min_{j\in [m]} \mu_j(x_i) \given 
        \substack{
            p_i\leq \max_{j\in [m]}{\mu_j(x_i)}, \\
            \forall i'<i,~ p_{i'} > \max_{j\in [m]}{\mu_j(x_{i'})}
        }
        }
        =\\ 
        \sum_{i\geq 1} \P*{\min_{j\in [m]}{i^*_j} = i}\cdot \Tag<sigconf>{\\} \P*{p_i> \min_{j\in [m]} \mu_j(x_i) \given p_i\leq \max_{j\in [m]}{\mu_j(x_i)}}
        =\\
        \sum_{i\geq 1} \P*{\min_{j\in [m]}{i^*_j} = i}\cdot \Tag<sigconf>{\\}\frac{\P*{\min_{j\in [m]} \mu_j(x_i)<p_i\leq \max_{j\in [m]}{\mu_j(x_i)}}}{\P*{p_i\leq \max_{j\in [m]}{\mu_j(x_i)}}}
        =\\
        \sum_{i\geq 1} \P*{\min_{j\in [m]}{i^*_j} = i}\cdot \Tag<sigconf>{\\} \frac{q^{-1} \cdot \sum_{x\in [q]} \parens{\max_{j\in [m]} \mu_j(x)-\min_{j\in [m]} \mu_j(x)}}{q^{-1} \cdot \sum_{x\in [q]} \max_{j\in [m]}{\mu_j(x)}}
        =\\
        \frac{\sum_{x\in [q]} \parens{\max_{j\in [m]} \mu_j(x)-\min_{j\in [m]} \mu_j(x)}}{\sum_{x\in [q]} \max_{j\in [m]}{\mu_j(x)}}.
        \qedhere
    \end{multline*}
\end{proof}

We additionally show that a popular sampling strategy called the ``Gumbel trick'', used widely in machine learning and in particular in the context of autoregressive models \cite[see, e.g.,][]{JGP16}, is robust and satisfies the same bound as \cref{lem:uni-coupler-gurantee}; thus it can be used as an alternative to \textsc{MinCoupler}. The algorithm, described in \cref{alg:uni-coupling-Gumbel}, is based on the well-known property of exponential random variables that if $X_1\sim \operatorname{Exp}(m_1),\dots,X_q\sim \operatorname{Exp}(m_q)$ are independent, then the probability that $X_i$ is the smallest among $X_1,\dots,X_q$ is exactly $m_i/(m_1+\dots+m_q)$.

\begin{Algorithm}
\caption{Universal coupler \textsc{GumbelTrick} \cite[see, e.g.,][]{JGP16}}
\label{alg:uni-coupling-Gumbel}
\KwIn{Distribution $\mu \in \Delta_q$, randomness $r$}
\KwOut{A sample from $\mu$}
Interpret $r$ as $q$ i.i.d.\ exponential random variables $r_1, \dots, r_q\sim \operatorname{Exp}(1)$\;
\Return{$\arg\min\set{r_x/\mu(x)\given x\in [q]}$}
\end{Algorithm}

\begin{remark}
	The Gumbel trick is often presented in a syntactically different form. Most often in practice, the distribution $\mu$ is given as input by the log-likelihoods $\log(\mu)$. Instead of using $r_1,\dots, r_q$, one uses random variables $\gamma_1,\dots,\gamma_n$ that follow the Gumbel distribution. The algorithm returns the $x$ that maximizes $\log(\mu(x))+\gamma_x$. The Gumbel distribution can be most simply defined as the distribution of $-\log(X)$ for $X\sim \operatorname{Exp}(1)$. The equivalence of \cref{alg:uni-coupling-Gumbel} immediately follows. 
\end{remark}

\begin{lemma}[robustness of \textsc{GumbelTrick}]
	\label{lem:uni-coupler-Gumbel-gurantee}
    Consider any $m,q>0$.
    Suppose $r=(r_1,\dots,r_q)$ is chosen by independently sampling each $r_i\sim \operatorname{Exp}(1)$.
    For any distributions $\mu_1,\cdots, \mu_m\in \Delta_q$, the probability that there exist $i,j\in [m]$ such that $ \textsc{GumbelTrick}(\mu_i,r) \neq \textsc{GumbelTrick}(\mu_j,r)$ is at most 
    \begin{align*}
        \frac{\sum_{x\in [q]} \parens{\max_{i\in [m]}{\mu_i(x)} - \min_{i\in [m]}{\mu_i(x)}}}{\sum_{x\in [q]} \max_{i\in [m]}{\mu_i(x)}}.
    \end{align*}
\end{lemma}

\begin{proof}
	Define $\mu_{\min}(x)=\min\set{\mu_i(x)\given x\in [q]}$ and $\mu_{\max}(x)=\max\set{\mu_i(x)\given x\in [q]}$. For each $y\in [q]$, let $\mathcal{E}_y$ be the event that $r_y/\mu_{\min}(y)<r_x/\mu_{\max}(x)$ for all $x\in [q]-\set{y}$. Conditioned on $\mathcal{E}_y$, it is easy to see that $\textsc{GumbelTrick}(\mu_i, r)=y$ for all $i$. This is because for any $x\in [q]-\set{y}$
	\[r_y/\mu_i(y)\leq r_y/\mu_{\min}(y)< r_x/\mu_{\max}(x)\leq r_x/\mu_i(x),\]
	which means $y=\arg\min\set{r_x/\mu(x)\given x\in [q]}$. Note also that this implies the events $\mathcal{E}_y$ are disjoint for $y\in [q]$. So we can upper bound the probability that there is a pair $i,j$ with $\textsc{GumbelTrick}(\mu_i, r)\neq \textsc{GumbelTrick}(\mu_j, r)$ by
	\[ 1-\P*{\bigcup_{y\in [q]} \mathcal{E}_y}=1-\sum_{y\in [q]}\P{\mathcal{E}_y}. \]
	Now notice that $r_x/\mu_{\max}(x)\sim \operatorname{Exp}(\mu_{\max}(x))$ and $r_y/\mu_{\min}(y)\sim \operatorname{Exp}(\mu_{\min}(y))$. So $\mathcal{E}_y$ is the probability that one exponential random variable is the smallest among several independent ones (with different rates), which is proportional to the rate of the exponential random variable:
	\[ \P{\mathcal{E}_y}=\frac{\mu_{\min}(y)}{\mu_{\min}(y)+\sum_{x\in [q]-\set{y}}\mu_{\max}(x)}\geq \frac{\mu_{\min}(y)}{\sum_{x\in [q]}\mu_{\max}(x)}. \]
	We conclude by calculating
	\[ 1-\sum_{y\in [q]} \frac{\mu_{\min}(y)}{\sum_{x\in [q]}\mu_{\max}(x)} = \frac{\sum_{x\in [q]} \parens{\max_{i\in [m]}{\mu_i(x)} - \min_{i\in [m]}{\mu_i(x)}}}{\sum_{x\in [q]} \max_{i\in [m]}{\mu_i(x)}}. \qedhere \]
\end{proof}

Finally, we characterize the performance of a robust universal coupler on randomly constructed distributions $\mu_1,\cdots, \mu_m$ where $\braces{\mu_i(x)}_{i\in [m]}$ is a martingale for each $x\in [q]$. For these distributions, we can bound the chance of not-all-equal samples by an expression that is only related to the expected $\ell_1$ or $\ell_2$ distances between the first distribution $p_1$ and the last distribution $p_m$, with some $O(\log nq)$ or $O(\sqrt{q})$ factor blowup and some small additive terms.

\begin{lemma}
    \label{lem:uni-coupler-on-martingale}
    Consider any $m,q>0$ and randomly constructed distributions $\mu_1,\cdots, \mu_m\in \Delta_q$.
    Suppose for any $i\in[m], x\in [q]$, $\mu_i(x)$ is a random variable such that $\sum_{x\in [q]} \mu_i(x)=1$ for any $i\in [m]$. If for each $x\in [q]$, $\braces{\mu_i(x)}_{i\in [m]}$ forms a martingale, then for any $n>0$, $\sum_{x\in [q]} \E{\max_{i\in [m]} \mu_i(x) - \min_{i\in [m]} \mu_i(x)}$ is at most
    \Tag{
    \[
         \min\set*{O(\log nq) \cdot \E{\tv{\mu_1,\mu_m}} + \frac{1}{n}, O(\sqrt{q})\cdot  \E*{\sum_{x\in [q]}(\mu_m(x)-\mu_1(x))^2}^{1/2}}
    \]
    }%
    \Tag<sigconf>{
    \begin{multline*}
         \min\Biggl\{O(\log nq) \cdot \E{\tv{\mu_1,\mu_m}} + \frac{1}{n},\\ O(\sqrt{q})\cdot  \E*{\sum_{x\in [q]}(\mu_m(x)-\mu_1(x))^2}^{1/2}\Biggr\}
    \end{multline*}
    }
\end{lemma}
The proof of this lemma follows Doob's maximal inequality on $\ell_1$ and $\ell_2$ norms.
For any martinagle $Y_0,Y_1,\cdots, Y_m$, Doob's maximal inequality gives upper bound for the expected maximum differences between any $Y_i$ and $Y_0$ simply by the expected differences between $Y_m$ and $Y_0$.
\begin{lemma}[Doob's maximal inequality, {\cite[see, e.g., ][]{Rev13}}]
 \label{lem:doob-inequality}
    For any martingales $Y_0,Y_1,\cdots, Y_m$ and any $p\geq 1, C>0$, the complementary cumulative distribution function of $\max_{i\in [m]} \abs{Y_i-Y_0}$ satisfies
    \begin{align*}
        \P *{\max_{i\in [m]}{\abs{Y_i-Y_0}} \geq C} \leq
        \frac{\E*{\abs{Y_m-Y_0}^p}}{C^p},
    \end{align*}
    and for any $p>1$,
    \begin{align*}
        \E *{\max_{i\in [m]}{\abs{Y_i-Y_0}^p}} \leq
        \parens*{\frac{p}{p-1}}^p \cdot \E*{\abs{Y_m-Y_0}^p}.
    \end{align*}
\end{lemma}
\begin{corollary}
    \label{cor:doob-inequality}
    For any martingales $Y_0,Y_1,\cdots, Y_m$ on support $[0,1]$ and any $N>0$,
    \begin{align*}
        \E *{\max_{i\in [m]}{\abs{Y_i-Y_0}}} \leq
        O(\log N) \cdot \E*{\abs{Y_m-Y_0}} + \frac{1}{N}.
    \end{align*}
\end{corollary}
\begin{proof}
    Taking $p=1$ in \cref{lem:doob-inequality}, we have \Tag{$\P{\max_{i\in [m]}{\abs{Y_i-Y_0}}\geq C} \leq \E{\abs{Y_m-Y_0}}/C$}\Tag<sigconf>{\[\P{\max_{i\in [m]}{\abs{Y_i-Y_0}}\geq C} \leq \E{\abs{Y_m-Y_0}}/C\]} for any $C>0$. 
    Therefore, for any $N>0$,
    \begin{multline*}
        \E*{\max_{i\in [m]}{\abs{Y_i-Y_0}}} 
        = 
        \int_0^1 \P*{\max_{i\in [m]} \abs{Y_i-Y_0}\geq C} dC
        \leq\Tag<sigconf>{\\}
        \frac{1}{N}
        +
        \int_{1/N}^1 \P*{\max_{i\in [m]} \abs{Y_i-Y_0}\geq C} dC
        \leq\\
        \frac{1}{N}
        +
        \int_{1/N}^1 \frac{\E*{\abs{Y_m-Y_0}}}{C}  dC
        =\Tag<sigconf>{\\}
        O(\log N) \cdot \E*{\abs{Y_m-Y_0}} + \frac{1}{N}.
        \qedhere
    \end{multline*}
\end{proof}

\begin{proof}[Proof of \cref{lem:uni-coupler-on-martingale}]
    Using the linearity of expectation, we get
    \begin{multline}\label{eqn:err-of-uni-couple}
        \sum_{x\in [q]} \E*{\max_{i\in [m]}{\mu_i(x)} - 
        \min_{i\in [m]}{\mu_i(x)} 
        }
        =\Tag<sigconf>{\\}
        \sum_{x\in [q]} \E*{\max_{i\in [m]}{\mu_i(x)} -\mu_1(x)}
        + 
        \E*{
        \mu_1(x) - \min_{i\in [m]}{\mu_i(x)}}
        \leq\\
        2
        \sum_{x\in [q]} \E*{\max_{i\in [m]}{\abs{\mu_i(x)-\mu_1(x)}}}
    \end{multline}
    Using \cref{cor:doob-inequality}, we can obtain the first upper bound\Tag{ for $\sum_{x\in [q]} \E{\max_{i\in [m]} \mu_i(x) - \min_{i\in [m]} \mu_i(x)}$}:
    \Tag{
    \[
        \cref{eqn:err-of-uni-couple} \leq \sum_{x\in [q]}
        O(\log nq) \cdot \E*{\abs{\mu_m(x)-\mu_1(x)}} + \frac{1}{nq}
        =\\
        O(\log nq) \cdot \E*{\tv{\mu_1(x),\mu_m(x)}} + 
        \frac{1}{n}.
    \]
    }%
    \Tag<sigconf>{
    \begin{multline*}
        \cref{eqn:err-of-uni-couple} \leq \sum_{x\in [q]}
        O(\log nq) \cdot \E*{\abs{\mu_m(x)-\mu_1(x)}} + \frac{1}{nq}
        =\\
        O(\log nq) \cdot \E*{\tv{\mu_1(x),\mu_m(x)}} + 
        \frac{1}{n}.
    \end{multline*}}%
    Using Cauchy-Schwarz inequality and Doob's maximal inequality for $\ell_2$ norm, we can obtain the second upper bound\Tag{ for $\sum_{x\in [q]} \E{\max_{i\in [m]} \mu_i(x) - \min_{i\in [m]} \mu_i(x)}$}:
    \begin{multline*}
        \cref{eqn:err-of-uni-couple}
        =
        2 
         \E*{ \sum_{x\in [q]} \max_{i\in[m]} (\mu_i(x)-\mu_1(x))}
        \leq\Tag<sigconf>{\\}
        2
         \E*{ \parens*{\sum_{x\in [q]} \max_{i\in[m]} (\mu_i(x)-\mu_1(x))}^2 }^{1/2}
        \leq\\
        2\E*{ q\sum_{x\in [q]} \max_{i\in[m]} (\mu_i(x)-\mu_1(x))^2 }^{1/2}
        \leq\Tag<sigconf>{\\}
        O(\sqrt{q}) \cdot 
        \E*{ \sum_{x\in [q]} (\mu_m(x)-\mu_1(x))^2 }^{1/2}. \qedhere
    \end{multline*}
\end{proof}

%% file: algorithm.tex
\section{Sublinear Parallel Sampling via Counting Oracles}\label{sec:alg}

In this section, we show our main \cref{thm:main} that we can (approximately) sample from a distribution after a sublinear number of rounds (in terms of the number of variables) of querying a polynomial number of the distribution's counting oracles.

\subsection{Algorithm}

We present our algorithm in \cref{alg:sample-on-hypergrid2}.
Let $u_1,\dots,u_n$ be the random seeds of the algorithm. 
The algorithm also shuffles the coordinates with a uniformly random permutation that we ignore in this description for simplicity. 

As a subroutine, we use a universal coupler \textsc{UniversalCoupler} that is robust, see \cref{def:robust-universal-coupler}. For instance, this can be either the \textsc{MinCoupler} (\cref{alg:uni-coupling}), or the \textsc{GumbelTrick} (\cref{alg:uni-coupling-Gumbel}).

We let $x \in [q]^n$ denote a sample from the target distribution generated using the naive sequential algorithm that iteratively samples the $i$-th entry conditioning on all previous entries: 
$$x_{i} \gets \textsc{UniversalCoupler}\parens*{X_{i} \given \braces*{X_{j}=x_{j}}_{j\in [i-1]}, u_i}.$$

The goal of the algorithm is to sample faster than one coordinate per iteration. The algorithm maintains an index $a$ where the $a$-th and earlier entries are all correctly sampled. At the $t$-th iteration, the algorithm attempts to resample all entries after $a$ by conditioning on the $a$-th and earlier entries: 
$$x^{t}_{i} \gets \textsc{UniversalCoupler}\parens*{X_{i} \given \braces*{X_{j}=x^{t-1}_{j}}_{j\in [a]}, u}.$$

Then, the algorithm uses $x^{t}$ to find the earliest entry $a'$ where sampling conditioning on $x^{t}_{j \in [a'-1]}$ differs from sampling conditioning on $x^{t-1}_{j \in [a]}$ and immediately fixes the entry $a'$: then $a'$ is a new index where the $a'$-th and earlier entries are all correctly sampled.

\begin{Algorithm}
\caption{Parallel sampling on product spaces}
\label{alg:sample-on-hypergrid2}
\KwIn{Counting oracle $\mu$, robust universal coupler \textsc{UniversalCoupler} on $[q]$}
\KwOut{A sample in $[q]^n$}
Sample a permutation $\sigma\gets \operatorname{uniform}(\mathcal{S}_n)$ \;
Sample i.i.d.\ random sources $u_1,u_2,\cdots, u_n\gets \operatorname{uniform}([0,1])$ \;
Initialize $a\gets 0, t\gets 0, x^0\gets\operatorname{null}$ \;
\While{\bf{true}}{
    $t\gets t+1$ \;
    \For{$i\in [n]$ \bf{in parallel}}{
        $y^t_{\sigma(i)} \gets \textsc{UniversalCoupler}\parens*{X_{\sigma(i)} \given \braces*{X_{\sigma(j)}=x^{t-1}_{\sigma(j)}}_{j\in [a]}, u_i}$ \;
    }
    \For{$i\in [n]$ \bf{in parallel}}{
        $x^{t}_{\sigma(i)} \gets \textsc{UniversalCoupler}\parens*{X_{\sigma(i)} \given \braces*{X_{\sigma(j)}=y^{t}_{\sigma(j)}}_{j\in [i-1]}, u_i}$ \;
    }
    \If{$x^t=y^t$}{
        \KwRet{$x^t$}
    }
    $a\gets \min \braces *{i\in [n]\given  y^t_{\sigma(i)} \neq x^t_{\sigma(i)}}$ \;
    \If{$a=n$}{
        \KwRet{$x^t$}
    }
}
\end{Algorithm}

\subsection{Correctness}

We consider a function $\tilde{x}:\mathcal{S}_n\times [0,1]^n \to [q]^n$ defined iteratively as follows:
\begin{align}
    \label{eqn:output-of-our-algo}
    \tilde{x}_i(\sigma,u) 
    = 
    \textsc{UniversalCoupler} \parens *{
        X_{\sigma(i)}
            \given
        \braces*{
            X_{\sigma(j)} = \tilde{x}_j(\sigma,u)
        }_{
            j\in [i-1]
        }
        ,
        u_i
    }.
\end{align}

For each $i\in [n]$, because \textsc{UniversalCoupler} is a universal coupler, $\tilde{x}_i(\sigma,u)$ follows the marginal distribution of $X_{\sigma(i)}$ conditioning on $X_{\sigma(1)}=\tilde{x}_1(\sigma,u), X_{\sigma(2)}=\tilde{x}_2(\sigma,u), \cdots, X_{\sigma(i-1)}=\tilde{x}_{i-1}(\sigma,u)$ (considering only the randomness of $u_i$).
Therefore, this function can serve as an objective output of the algorithm when we have fixed the randomness $\sigma$ and $u$.
\begin{lemma}
\label{lem:obj-output}
    If \cref{alg:sample-on-hypergrid2} always outputs $x^t_{\sigma(i)}=\tilde{x}_i(\sigma,u)$, it samples perfectly from the distribution of $X$.
\end{lemma}

Let $a^t$ be the value of $a$ at the end of round $t$ (if the algorithm does not terminate with $x^t=y^t$ before updating the value of $a$ in round $t$).
For simplicity, we suppose $a^0=0$.
Next, we show that the vector $x^t$ produced by the algorithm in each round matches this objective vector in the first $a^t$ entries.
\begin{lemma}
    \label{lem:fixed-after-a}
    For any $\sigma, u$, after each round $t$ of \cref{alg:sample-on-hypergrid2},
    \begin{align*}
        \forall i\in [a^t], \quad 
        x^t_{\sigma(i)} 
            = 
        \tilde{x}_i (\sigma, u).
    \end{align*}
\end{lemma}
\begin{proof}
    According to the definition of $a^t$, we have $\forall i\in [a^t-1]$, $x^t_{\sigma(i)}=y^t_{\sigma(i)}$. 
    Therefore, according to the definition of $x^t$, we have for any $i\in [a^t]$, 
    \begin{align}
        \label{eqn:same-recrusion-for-xt}
        x^t_{\sigma(i)} = \textsc{UniversalCoupler}\parens*{X_{\sigma(i)} \given \braces*{X_{\sigma(j)}=x^t_{\sigma(j)}}_{j\in [i-1]},~u_i}.
    \end{align}
    Note that this recursion matches the recursion \cref{eqn:output-of-our-algo} used in the definition of $\tilde{x}_i(\sigma, u)$ for any $i\in [a^t]$.
    Therefore, $\forall i\in [a^t], x^t_{\sigma(i)} = \tilde{x}_i(\sigma,u)$.
\end{proof}

To show that \cref{alg:sample-on-hypergrid2} is making progress every iteration, we prove $a^t$ is (strictly) monotone in terms of $t$.

\begin{lemma}
\label{lem:at-increasing}
    $a^t$ is strictly increasing with $t$.
\end{lemma}
\begin{proof}
    According to the definition of $y^t$, for any $i\in [a^{t-1}]$, $y^t_{\sigma(i)}=x^{t-1}_{\sigma(i)}=\tilde{x}_{i}(\sigma,u)$.
    Therefore, according to the definition of $x^t$, for any $i\in [a^{t-1}+1]$, 
    \begin{align*}
        x^t_{\sigma(i)} 
        &= 
        \textsc{UniversalCoupler} \parens*{X_{\sigma(i)}\given \braces*{X_{\sigma(j)}=x^{t-1}_{\sigma(j)}}_{j\in [i-1]}}
        \\
        &=
        \textsc{UniversalCoupler} \parens*{X_{\sigma(i)}\given \braces*{X_{\sigma(j)}=\tilde{x}_j(\sigma, u)}_{j\in [i-1]}}
        \tag{definition of $a^{t-1}$}
        \\
        &=
        \tilde{x}_i(\sigma, u)
        =
        y^t_{\sigma(i)}.
        \tag{definition of $\tilde{x}_i(\sigma, u)$}
    \end{align*}
    By the definition of $a^t$, if $x^t\neq y^t$, we get $a^t>a^{t-1}+1$.
\end{proof}

Note that the algorithm terminates when either of the following two conditions is satisfied in some round $t$: $a^t=n$ or $x^t=y^t$.
Due to \cref{lem:at-increasing}, the algorithm always terminates. 
If it terminates because of the first condition $a^t=n$, due to \cref{lem:fixed-after-a}, the output of the algorithm matches the objective in \cref{lem:obj-output}.
Otherwise, because of the definition of $x^t$ and the fact that $x^t=y^t$, the output satisfies \cref{eqn:same-recrusion-for-xt}, which matches the recursion \cref{eqn:output-of-our-algo} used in the definition of $\tilde{x}_i(\sigma,u)$, and thus matches the objective in \cref{lem:obj-output}.
As a conclusion, we obtain the correctness of our algorithm.
\begin{lemma}
\cref{alg:sample-on-hypergrid2} returns a sample $x\sim \mu$.
\end{lemma}

\subsection{Round Complexity}

We establish our sublinear round complexity via two steps. 
First, we ignore the randomness of $\sigma$ and $u$, and establish a worst-case round complexity, which can be linear with some choices of $\sigma,u$.
Second, we show that the expectation of this round complexity is actually $\widetilde{O}(n^{2/3}\log q)$ with the random choices of $\sigma, u$.
This bound is established by the robustness of the universal coupler on randomly constructed distributions that satisfy the martingale property (\cref{def:robust-universal-coupler} and \cref{lem:uni-coupler-on-martingale}) and a pinning lemma (\cref{cor:pinning-on-sqr-prob-diff}).

To use this algorithm to nontrivially speed up sampling of \emph{planar perfect matchings}, we also need a tail bound for the round complexity. 
Because of Markov's inequality, the expected round complexity implies a simple tail bound -- the round complexity is less than $c\cdot n^{2/3}\log q$ with probability $\widetilde{\Omega}(1/c)$.
At the end of this subsection, we boost this tail bound to $2^{-\widetilde{\Omega}(c)}$ via the simple observation that running several rounds of the algorithm is equivalent to reinitiating the algorithm with a smaller instance. 

\paragraph{Worst-case round complexity.}
First, we consider the randomness $\sigma, u$ used in the algorithm as part of the input, and give an upper bound for the round complexity.
For each $\sigma\in \mathcal{S}_n, u\in [0,1]^n$ and $i\in [n]$, let $\overline{a}_i(\sigma,u)$ be the maximum $a$ such that \cref{alg:sample-on-hypergrid2} will not correctly sample $X_{\sigma(i)}$ to $\tilde{x}_i(\sigma,u)$, the value of $X_{\sigma(i)}$ in the final output, under the correct conditioning of $X_{\sigma(1)},\dots, X_{\sigma(a)}$.
Formally, $\tilde{a}_i(\sigma, u)$ is defined as follows:
\Tag{\[
    \overline{a}_i(\sigma,u)
    =
    \max
    \braces *{
        a\geq 0 \given
        \tilde{x}_i(\sigma, u) 
        \neq 
        \textsc{UniversalCoupler} \parens*{X_{\sigma(i)} \given \braces*{X_{\sigma(j)}=\tilde{x}_j(\sigma, u)}_{j \in [a]}, u_i}
    },
\]}%
\Tag<sigconf>{
\begin{multline*}
    \overline{a}_i(\sigma,u)
    =
    \max
    \Bigl\{
        a\geq 0 \;\Big|\;
        \tilde{x}_i(\sigma, u) 
        \neq \\
        \textsc{UniversalCoupler} \parens*{X_{\sigma(i)} \given \braces*{X_{\sigma(j)}=\tilde{x}_j(\sigma, u)}_{j \in [a]}, u_i}
    \Bigr\},
\end{multline*}
}%
where we define the maximum of an empty set to be $0$ for simplicity. 
Using this definition, we can establish a worst-case round complexity of \cref{alg:sample-on-hypergrid2}.

\begin{lemma}
    \label{lem:det-upper-bound-for-rounds}
    For any integer $\theta\geq 1$ and randomness $\sigma\in \mathcal{S}_n, u\in [0,1]^n$, the round complexity of \cref{alg:sample-on-hypergrid2} is at most 
    \[
        \abs *{
            \braces *{
                i\in [n]
                \given
                \overline{a}_i(\sigma,u) 
                \geq  
                i-\theta
            }
        }
        +
        1
        +
        \frac{n}{\theta}.
    \]
 \end{lemma}
 \begin{proof}
    Recall that we define $a^t$ as the value of $a$ after round $t$,
     and we define $a^0$ as $0$.
     The algorithm has at most one round $t$ without computing $a^t$, when it terminates with $x^t=y^t$.
     Suppose that the step size of any round $t$, where the algorithm computes $a^t$, is the increment $a^t-a^{t-1}$. 
     Based on the step sizes, we divide the rounds into two classes:
         \textbf{small-progress rounds} that have step sizes $< \theta$, and 
         \textbf{large-progress rounds} that have step sizes $\geq \theta$.
     Note that the number of large-progress rounds is at most $n/\theta$ because otherwise $a$ will exceed $n$ in some round.
     It suffices to upper bound the number of small-progress rounds by $\abs *{
            \braces{
                i\in [n]
                :
                \overline{a}_i(\sigma,u) 
                \geq  
                i - \theta
            }
        }$ to finish the proof.

        Consider any round $t$ such that $a^{t}-a^{t-1} < \theta$. 
        Due to the definition of the algorithm and \cref{lem:fixed-after-a}, the algorithm finds $y^t_{\sigma(i)}=x^t_{\sigma(i)} = \tilde{x}_i(\sigma, u)$ for any $i<a^t$.
        The algorithm also finds $x^t_{\sigma(a^t)} \neq y^t_{\sigma(a^t)}$. 
        According to the definition of $x^t$ and $y^t$ and \cref{lem:fixed-after-a}, $y^t_{\sigma(a^t)}=$
        \[
            \textsc{UniversalCoupler} \parens*{X_{\sigma(a^t+1)}\given \braces*{X_{\sigma(j)}=\tilde{x}_j(\sigma, u)}_{j\in [a^{t-1}]}, u_{a^t}},
        \]
        and $x^t_{\sigma(a^t)}=$
        \Tag{
        \[
            \textsc{UniversalCoupler} \parens*{X_{\sigma(a^t+1)}\given \braces*{X_{\sigma(j)}=\tilde{x}_j(\sigma, u)}_{j\in [a^t-1]}, u_{a^t}}
            =
            \tilde{x}_{a^t}(\sigma, u).
        \]
        }%
        \Tag<sigconf>{
		\begin{multline*}
            \textsc{UniversalCoupler} \parens*{X_{\sigma(a^t+1)}\given \braces*{X_{\sigma(j)}=\tilde{x}_j(\sigma, u)}_{j\in [a^t-1]}, u_{a^t}}
            =\\
            \tilde{x}_{a^t}(\sigma, u).
        \end{multline*}}%
        This implies
        $\overline{a}_{a^t}(\sigma, u) \geq a^{t-1} > a^t-\theta$.
        Therefore, $a^t\in \braces{
                i\in [n]
                \given
                \overline{a}_i(\sigma,u) 
                \geq  
                i - \theta
            }$.
        Since $a^t$ are strictly increasing (\cref{lem:at-increasing}), the number of such small-progress rounds can be upper bounded by $\card*{
            \braces{
                i\in [n]
                \given
                \overline{a}_i(\sigma,u) 
                \geq  
                i - \theta
            }
        }$.
 \end{proof}

We note that the cardinality of the set $\braces {i\in [n] \given \overline{a}_i(\sigma,u) \geq i-\theta}$ can be $\Omega(n)$, even under expectation over $u$. 
Suppose that $X_1,X_3,\dots, X_{n-1}$ are sampled independently and uniformly at random, and $X_2=X_1, X_4=X_3, \cdots, X_n=X_{n-1}$. 
For $\theta \geq 2$ and the permutation $\sigma(i)=i$, this set will involve each $2i$ with probability $1/2$ independently. 

 \begin{remark}
    Because our analysis only needs the {\bf large-progress rounds} to increase $a$ by $\geq \theta$, we could still enjoy the same upper bound if we only resample the first $\theta$ entries after $a$ in each iteration. This suggests a more query-efficient implementation \cref{alg:sample-on-hypergrid-efficient}.
 \end{remark}

 \begin{Algorithm}
\caption{Query-efficient implementation for each iteration of \cref{alg:sample-on-hypergrid2}}
\label{alg:sample-on-hypergrid-efficient}
\While{\bf{true}}{
    $t\gets t+1$ \;
    $y^t\gets x^{t-1}$\;
    \For{$i\in \braces*{a+1,\dots, \min\{a+\theta,n\} }$ \bf{in parallel}}{
        $y^t_{\sigma(i)} \gets \textsc{UniversalCoupler}\parens*{X_{\sigma(i)} \given \braces*{X_{\sigma(j)}=x^{t-1}_{\sigma(j)}}_{j\in [a]}, u_i}$ \;
    }
    \For{$i\in \braces*{a+1,\dots, \min\{a+\theta,n\} }$ \bf{in parallel}}{
        $x^{t}_{\sigma(i)} \gets \textsc{UniversalCoupler}\parens*{X_{\sigma(i)} \given \braces*{X_{\sigma(j)}=y^{t}_{\sigma(j)}}_{j\in [i-1]}, u_i}$ \;
    }
    $a\gets \min ~\braces *{i\in \braces*{a+1,\dots,\min\{a+\theta,n\} }: y^t_{\sigma(i)} \neq x^t_{\sigma(i)}}\cup \braces{\min\{a+\theta,n\}+1} $ \;
    \If{$a\geq n$}{
        \KwRet{$x^t$}
    }
}
\end{Algorithm}

\paragraph{Expected round complexity.} Next, we consider the randomness $\sigma, u$ and show an improved expected round complexity for the algorithm. Based on the worst-case analysis, we give a better bound for the expected cardinality of $\braces{ i\in [n] \given \overline{a}_i(\sigma,u) \geq i - \theta }$ for some sophisticated choice of $\theta$. Due to the linearity of expectation, we can separately upper bound the probability of $\overline{a}_i(\sigma, u)\geq i-\theta$ for each $i\in [n]$. Only considering the randomness of $u$, we can obtain the following lemma, which upper bounds the probability by the total variation distance between the conditional distributions of $X_{\sigma(i)}$ under fully conditioning of all previous variables ($X_{\sigma(1)}$ to $X_{\sigma(i-1)}$) and partial conditioning of some previous variables ($X_{\sigma(1)}$ to $X_{\sigma(i-\theta)}$).

\begin{lemma}
    \label{lem:tv-bound-for-large-overline-a-grid}
    For any $i\geq \theta$ and any permutation $\sigma\in \mathcal{S}_n$, the probability of $\overline{a}_i(\sigma,u)\geq i-\theta$ is at most
     \Tag{\[
     O(\min{\braces*{\log nq, \sqrt{q}}}) 
     \cdot 
     \E *_{X}{\tv*{X_{\sigma(i)}\given \braces*{X_{\sigma(j)}}_{j\in [i-1]}, X_{\sigma(i)} \given \braces*{X_{\sigma(j)}}_{j\in [i-\theta]}}^2}^{1/2} 
     + 
     \frac{1}{n},
     \]}%
     \Tag<sigconf>{
     \begin{multline*}
     \frac{1}{n}+O(\min{\braces*{\log nq, \sqrt{q}}}) 
     \cdot \\
     \E *_{X}{\tv*{X_{\sigma(i)}\given \braces*{X_{\sigma(j)}}_{j\in [i-1]}, X_{\sigma(i)} \given \braces*{X_{\sigma(j)}}_{j\in [i-\theta]}}^2}^{1/2} 
     ,
     \end{multline*}
     }%
     where the probability is taken only over the randomness of $u_1,\cdots, u_n$.
\end{lemma}
\begin{proof}
    According to the definition of $\overline{a}_i(\sigma,u)$, $\overline{a}_i(\sigma,u)\geq i-\theta$ when there exists $a\in [i-\theta, i-1]$ such that 
    $
        \textsc{UniversalCoupler} \parens{X_{\sigma(i)} \given \braces{X_{\sigma(j)}=\tilde{x}_j(\sigma, u)}_{j \in [i-1]}, u_i}
        \neq 
        \textsc{UniversalCoupler} \parens{X_{\sigma(i)} \given \braces{X_{\sigma(j)}=\tilde{x}_j(\sigma, u)}_{j \in [a]}, u_i}
    $.
    Equivalently, when we apply the universal coupler on all variables $X_i | \braces{X_{\sigma(j)}}_{j\in [a]}$, where $a$ can be any integer in $[i-\theta, i-1]$, the coupler produces different outcomes for two of them. 
    Let $\braces{\mu_a(x)}_{x\in [q]}$ denote the variables characterizing the randomly constructed distribution of $X_{\sigma(i)} | \braces{X_{\sigma(j)}}_{j\in [a]}$, where the randomness is taken over the random conditioning of $\braces{X_{\sigma(j)}}_{j\in [a]}$, i.e.,
    \[
    \forall x\in [q], ~\mu_a(x) = \P *{X_{\sigma(i)}=x \given \braces{X_{\sigma(j)}}_{j\in [a]}}
    \]
    We have that $\P_{u}{\overline{a}_i(\sigma, u)\geq i-\theta} = \P_{u}{\exists a,a'\in [i-\theta,i-1], \textsc{UniversalCoupler}(\mu_a,u_i)\neq  \textsc{UniversalCoupler}(\mu_{a'},u_i)}$.
    Further, note that $\{\mu_a(x)\}_{a\in [i-\theta,i-1]}$ forms a martingale, where randomness is taken over $u_1,\cdots, u_{i-1}$.
    Because of the robustness of \textsc{UniversalCoupler}, \cref{def:robust-universal-coupler}, and \cref{lem:uni-coupler-on-martingale}, we can upper bound $
        \P *_{u}{\overline{a}_i(\sigma,u)\geq i-\theta}$ by the following terms:
    \begin{multline*}
        \E*{\sum_{x\in [q]} \max_{a\in [i-\theta, i-1]}{\mu_a(x)}-\min_{a\in [i-\theta,i-1]} \mu_a(x)}
        \leq\\
        \Tag{
        	\min\braces*{O(\log nq) \cdot \E{\tv{\mu_{i-\theta},\mu_{i-1}}} + \frac{1}{n}, O(\sqrt{q})\cdot  \E*{\sum_{x\in [q]}(\mu_{i-1}(x)-\mu_{i-\theta}(x))^2}^{1/2}}
        }
        \Tag<sigconf>{
        	\min\Biggl\{O(\log nq) \cdot \E{\tv{\mu_{i-\theta},\mu_{i-1}}} + \frac{1}{n},\\ O(\sqrt{q})\cdot  \E*{\sum_{x\in [q]}(\mu_{i-1}(x)-\mu_{i-\theta}(x))^2}^{1/2}\Biggr\}
        }
        \leq\\
        O(\min{\braces*{\log nq, \sqrt{q}}}) 
     \cdot 
     \E *_{X}{\tv*{\mu_{i-1}, \mu_{i-\theta}}^2}^{1/2} 
     + 
     \frac{1}{n}.
     \qedhere
    \end{multline*}
\end{proof}

Then, applying the pinning lemma (\cref{cor:pinning-on-sqr-prob-diff}), we can obtain a better average round complexity for \cref{alg:sample-on-hypergrid2}.
\begin{proof}[Proof of \cref{thm:main}]
    We show that the expected number of rounds of \cref{alg:sample-on-hypergrid2} is 
        \[            
            O(n^{2/3}\cdot \min\{\log^{2/3}{n}\log{q},q^{1/3}\log^{1/3}{q}\})
        \]
    Consider $\theta = O(n^{1/2})$. According to \cref{lem:det-upper-bound-for-rounds}, the expected number of rounds that \cref{alg:sample-on-hypergrid2} needs is at most 
    \begin{equation}
        \label{eqn:expected-running-time}
        \frac{n}{\theta} + \theta + 1 + \sum_{i\geq \theta} \P *_{\sigma, u}{\overline{a}_i(\sigma, u) \geq i-\theta} 
    \end{equation}

    By \cref{lem:tv-bound-for-large-overline-a-grid}, we can upper bound
    \begin{multline*}
	    \sum_{i\geq \theta} \P *_{\sigma, u}{\overline{a}_i(\sigma, u) \geq i-\theta}
	    \leq
        1+O(\min{\braces*{\log nq, \sqrt{q}}}) \cdot \\
        \underbrace{
            \sum_{i\geq \theta} 
            \E *{
            \tv*{X_{\sigma(i)}\given \braces*{X_{\sigma(j)}}_{j\in [i-1]}, X_{\sigma(i)} \given \braces*{X_{\sigma(j)}}_{j\in [i-\theta]}}^2
            }^{1/2}
        }_{
        },
    \end{multline*}
    where we can further upper bound the underbraced term by 
    \begin{multline*}
        \leq O(\sqrt{n} ) \cdot \\ \E *{
            \sum_{i\geq \theta} \tv*{X_{\sigma(i)}\given \braces*{X_{\sigma(j)}}_{j\in [i-1]}, X_{\sigma(i)} \given \braces*{X_{\sigma(j)}}_{j\in [i-\theta]}}
        }
        \leq\\
        O(\sqrt{n}) \cdot 
        \parens*{\frac{(\theta-1)\log q}{2}}
        ^{1/2}
        \leq
        O(\sqrt{n\theta\log{q}}) 
    \end{multline*}
    Therefore, if we take $\theta=\frac{n^{1/3}}{\log^{1/3}{q}(\min\braces{\log{nq},\sqrt{q}})^{2/3}}$, the expected round complexity of \cref{alg:sample-on-hypergrid2} can be upper bounded by
    \Tag{
    \[
        \frac{n}{\theta} + O(\sqrt{n\theta\log{q}}\min\{\log{nq},\sqrt{q}\}) \leq O(n^{2/3}\cdot \min\{\log^{2/3}{n}\log{q},q^{1/3}\log^{1/3}{q}\}).
    \]
    }%
    \Tag<sigconf>{
    \begin{multline*}
        \frac{n}{\theta} + O(\sqrt{n\theta\log{q}}\min\{\log{nq},\sqrt{q}\}) \leq \\ O(n^{2/3}\cdot \min\{\log^{2/3}{n}\log{q},q^{1/3}\log^{1/3}{q}\}).
    \end{multline*}
    }%
    Note that our choice of $\theta$ guarantees that $\sum_{i\geq \theta} \P_{\sigma, u}{\overline{a}_i(\sigma, u)\geq i-\theta} = O(n/\theta)$. The expected round complexity can also be upper bounded by $O(n/\theta)$.
    Using our query-efficient implementation, \cref{alg:sample-on-hypergrid-efficient}, the expected total number of queries we make is $O(n)$.
\end{proof}

\paragraph{Tail bounds for the round complexity.}
Because of Markov's inequality, we can obtain the following corollary:
\begin{corollary}
    \label{cor:rounds-wp-1/2}
    For any input $X_1,X_2, \cdots, X_n$, with probability at least $1/2$, \cref{alg:sample-on-hypergrid2} terminates in $\widetilde{O}(n^{2/3}\log q)$ rounds.
\end{corollary}

Next, we establish a tail bound for the number of rounds used by the algorithm.

\begin{theorem}\label{thm:tail-bounds}
    There exists a constant $M>0$ such that for any integer $c\geq 1$, \cref{alg:sample-on-hypergrid2} terminates in $cM \cdot (n\log n)^{2/3}\log q$ rounds with probability at least $1-2^{-c}$. 
\end{theorem}
\begin{proof}
    Let $M$ be a constant such that \cref{alg:sample-on-hypergrid2} terminates in $M \cdot (n\log n)^{2/3}\log q$ with probability at least $1/2$. The existence of such $M$ follows \cref{cor:rounds-wp-1/2}. Let $R$ be the variable that denotes the number of rounds \cref{alg:sample-on-hypergrid2} uses. Next, we prove by induction that for any integer $c\geq 1$, we have 
    \begin{align*}
        \P*{R>cM \cdot (n\log n)^{2/3}\log q} \leq 2^{-c}.
    \end{align*}
    The cases where $n=1$ or $c=1$ are trivial.

    Suppose that we have proved the theorem for any $n<N$. Consider an instance with $n=N$. For any $i\leq n-1$, let $\mathcal{A}_i$ be the event that \cref{alg:sample-on-hypergrid2} does not terminate and has $a=i$ after running it for $M \cdot (n\log n)^{2/3}\log q$ rounds. This definition immediately gives us $\sum_{i\leq n-1} \P{\mathcal{A}_i}\leq \frac{1}{2}$. Because of \cref{lem:at-increasing}, after running the algorithm for $M \cdot (n\log n)^{2/3}\log q$ rounds, we have $a\geq M \cdot (n\log n)^{2/3}\log q$. Therefore, $\P{\mathcal{A}_0}=0$ and for any $c\geq 2$,
    \begin{multline*}
        \P*{R>cM \cdot (n\log n)^{2/3}\log q} 
        = \\
        \sum_{i\in [n-1]} \P*{R>cM \cdot (n\log n)^{2/3}\log q \given \mathcal{A}_i} \P*{\mathcal{A}_i}
        \leq \\
        \frac{1}{2} \cdot \max_{i\in [n-1]} \P*{R>cM \cdot (n\log n)^{2/3}\log q \given \mathcal{A}_i}.
    \end{multline*}
    Next, we show $\P{R>cM \cdot (n\log n)^{2/3}\log q \given \mathcal{A}_i}\leq 2^{-c+1}$ for any $i\in [n-1]$ to finish the proof. Note that for any $i\in [n-1]$, whether the event $\mathcal{A}_i$ happens is determined by $\sigma(1),\cdots, \sigma(i)$ and $u_{1},\cdots, u_{i}$. Therefore, $\mathcal{A}_i$ is independent of the random permutation in $[n]\setminus \braces{\sigma(j)}_{j\in [i]}$ and the randomness of $u_{i+1},\cdots, u_n$. Further, if $\mathcal{A}_i$ happens after running the algorithm for $M \cdot (n\log n)^{2/3}\log q$ rounds, the algorithm fixes the values of the first $i$ variables and continues with $a=i$. Therefore, conditioning on any $\sigma(1),\cdots, \sigma(i)\in \binom{[n]}{i}$ and $u_1,\cdots, u_i\in [0,1]$ that cause $\mathcal{A}_i$ to happen, the remaining iterations of the algorithm are equivalent to those in a fresh run of the algorithm on the remaining variables $\braces{X_j}_{j\in [n]} \setminus \braces{X_{\sigma(j)}}_{j\in [i]}$ conditioning on $\braces{X_{\sigma(j)}=\tilde{x}_j(\sigma, u)}_{j\in [i]}$. According to our induction hypothesis, the number of remaining rounds of the algorithm is (strictly) greater than $(c-1)M \cdot ((n-i)\log (n-i))^{2/3}\log q$ with probability at most $2^{-c+1}$. Therefore, $\P{R>cM \cdot (n\log n)^{2/3}\given \mathcal{A}_i}$ is at most $2^{-c+1}$.
\end{proof}

\Tag{
\subsection{Sampling via Approximate Counting Oracles}\label{sec:approximate-oracle}
We show that our algorithm can also work with approximate counting oracles. 
Suppose $\hat{\mu}$ is an oracle such that for any $S\subseteq [n]$ and $y\in [q]^S$, with probability at least $1-\delta$,
\begin{align}
\label{eqn:approx-counting}
    \hat{\mu}(S,y) \in (1\pm \epsilon) \P_{X\sim \mu}{X_S=y}. 
\end{align}
At each round of the algorithm, we shall consider an approximate version of the conditional probability distribution. 
For any permutation $\sigma$, indices $0\leq a<i\leq n$ and any $y\in [q]^{\braces{\sigma(j)}_{j\in [a]}}$, we consider the following distribution $\upsilon$ for $X_{\sigma(i)} | \braces{X_{\sigma(j)}=y_{\sigma(j)}}_{j\in [a]}$: for any $x\in [q]$, let $y'(x)\in [q]^{\braces{\sigma(j)}_{j\in [a]} \cup \braces{\sigma(i)}}$ be the vector such that $y'_{\sigma(i)}(x)=x$ and $y'_{\sigma(j)}(x)=y_{\sigma(j)}$ for any $j\in [a]$, then we have
\begin{align*}
    \forall x\in [q], \quad \upsilon(x) \propto {\hat{\mu}\parens*{\braces{\sigma(j)}_{j\in [a]} \cup \braces{\sigma(i)}, y'(x)}}.
\end{align*}
In particular, we use $\upsilon_{i|a}(\sigma,u)$ to denote this approximate version of the distribution for $X_{\sigma(i)} | \braces{X_{\sigma(j)}=\tilde{x}_j(\sigma, u)}_{j\in [a]}$.
We say that a pair of randomness $(\sigma, u)$ is good if for any $0\leq a< i\leq n$, 
\Tag{
\begin{align*}
    \textsc{UniversalCoupler}\parens*{X_{\sigma(i)} \given \braces*{X_{\sigma(j)}=\tilde{x}_{j}(\sigma, u)}_{j\in [a]}, u_i} =  \textsc{UniversalCoupler}\parens*{\upsilon_{i|a}(\sigma, u), u_i}.
\end{align*}
}%
\Tag<sigconf>{
\begin{multline*}
    \textsc{UniversalCoupler}\parens*{X_{\sigma(i)} \given \braces*{X_{\sigma(j)}=\tilde{x}_{j}(\sigma, u)}_{j\in [a]}, u_i} =\\  \textsc{UniversalCoupler}\parens*{\upsilon_{i|a}(\sigma, u), u_i}.
\end{multline*}
}%

We show the following two key lemmas on good randomness. 
The first states that the approximate counting oracles do not influence the output of the algorithm as long as the randomness is good. 
The second upper bounds the probability that the randomness is not good.
\Tag{We defer the proofs to \cref{proof:approx-counting-equiv-under-good} and \cref{proof:approx-counting-good-prob}.}
\begin{lemma}
    \label{lem:approx-counting-equiv-under-good}
    If the pair of randomness $(\sigma,u)$ is good, \cref{alg:sample-on-hypergrid2} with the approximate counting oracle outputs the same vectors as \cref{alg:sample-on-hypergrid2} with the exact counting oracle in the same number of rounds.
\end{lemma}

\begin{lemma}
    \label{lem:approx-counting-good-prob}
    For any pair of randomness $(\sigma, u)$, it is good with probability at least $1-O(n^2\epsilon+n^2q\delta)$. 
\end{lemma}

Suppose the parameters of the approximate counting oracle satisfy $\delta, \epsilon = O(n^{-3}q^{-1})$.
Putting \cref{thm:tail-bounds} together with these two lemmas, we can easily obtain the guarantee: if we terminate ~\cref{alg:sample-on-hypergrid2} in $O(n^{2/3}\poly \log(n,q))$ rounds, the output distribution is within a total variation distance $O(n^{-1})$ of $\mu$.
}

%% file: application.tex
\section{Applications}\label{sec:applications}

In this section, we show an example application of \cref{thm:main}, to the problem of sampling uniformly random perfect matchings in planar graphs. The famous FKT algorithm allows parallel counting of the number of perfect matchings \cite[see, e.g.,][]{ABTV23}. The previous best parallel runtime for this problem is $\widetilde{O}(n^{1/2})$ for planar graphs of size $n$ \cite{ABTV23}. 

\begin{remark}
    Two key techniques for deterministic (approximate) counting, namely the tree recursion/correlation decay method \cite{Wei06} and the polynomial interpolation method \cite{Bar16} can often be trivially parallelized. The former involves solving a recursion on a tree of logarithmic depth, and the latter involves enumerating structures of logarithmic size in a host object (e.g., a graph). As such, our results automatically provide a parallel speedup wherever these methods apply.
\end{remark}

\begin{theorem}
    Let $G=(V, E)$ be a planar graph. There exists an algorithm that samples a uniformly random perfect matching in $G$ with a parallel runtime of $\widetilde{O}(n^{1/3})$ and $\poly(n)$ work.
\end{theorem}
\begin{proof}
    Similar to \cite{ABTV23}, we use the planar separator theorem to find a separator of size $O(\sqrt{n})$, sample the portion of the perfect matching incident to the separator, and then recursively sample the rest of the perfect matching in the now-disjoint halves of the graph, in parallel. Our modification is that, while na\"ively sampling the separator edges takes $\widetilde{O}(\sqrt{n})$ time, using \cref{thm:main}, we can speed it up to $\widetilde{O}(n^{1/3})$.

    To be more specific, given the input graph $G=(V, E)$, we find a planar separator $S\subseteq V$ of size $O(\sqrt{n})$, such that $G-S$ is composed of two smaller graphs, on vertex sets $A, B$, each of size $\leq (1-\Omega(1))n$. This can be done in parallel \cite{GM87}.

    Next, we consider the distribution $\mu$ on $E^S$, where $\mu(x)$ is proportional to the number of perfect matchings that have edge $x_v$ incident to $v$ for all $v\in S$. Note that many configurations $x\in E^S$ are invalid, for example, those where $v$ is not even an endpoint of $x_v$, or those with clashing edges for two vertices in $S$. All of these invalid configurations are assigned a measure of $0$ under $\mu$. We claim that there is a parallel (\Class{NC}) counting oracle for $\mu$. Indeed, given a partial pinning, we can check if it is valid, and if so, remove the edges in the pinning from the graph, and simply count perfect matchings in the resulting subgraph. The number of perfect matchings in planar graphs can be efficiently computed in parallel by the FKT algorithm \cite[see, e.g.,][]{ABTV23}.

    Now we use \cref{alg:sample-on-hypergrid2} to sample from $\mu$. Once the sample is produced, we remove all the endpoints of this partial matching from $G$ (in particular, this removes all of $S$), and now we have two disjoint subgraphs of geometrically smaller size. In parallel, we recurse on each.

    Note that the total number of calls to \cref{alg:sample-on-hypergrid2} is $\leq \poly(n)$. By using the tail bounds for our algorithm, \cref{thm:tail-bounds}, each call finishes in at most $\widetilde{O}(\sqrt{n}^{2/3})=\widetilde{O}(n^{1/3})$ time, with probability at least $1-1/\poly(n)$. Taking a union bound, and using the fact that recursively the subgraphs shrink geometrically, we get that the overall parallel runtime is $\widetilde{O}(n^{1/3})$ with high probability.
\end{proof}

%% file: hardness.tex
\section{Hardness}\label{sec:hardness}
    In this section, we prove that any algorithm cannot approximately sample within a constant total variation distance of arbitrary distribution $\mu$ with $n^{1/3-\Omega(1)}$ round complexity and a polynomial number of queries in each round to the exact counting oracle.
    More generally, we shall prove the following hardness result on parallel search via counting oracles for $q=2$.

    \begin{theorem}
        \label{thm:hardness-search}
        For any constant $\delta\in (0,1]$, any $c\in (0, n^{1-\delta})$ and any (randomized) algorithm $\ALG$ making at most $n^c$ queries to the counting oracle in each round, there exists an instance $\mu:\{0,1\}^n \to \{0,1\}$ such that $\ALG$ can only find a solution $x$ (such that $\mu(x)=1$) with probability at most $0.01$ after (strictly) less than $\frac{1}{4}\cdot \parens{\frac{n}{(c+2) \log n} }^{1/3}$ rounds of queries. 
    \end{theorem}

    In the rest of this section, we use $H=(S,y)$, where $S\subseteq [n]$ and $y\in [q]^S$, to denote a hypercube by $H=\braces{x\in [q]^n \given x_S=y}$.
    For the abuse of notation, for any function $\mu:\braces{0,1}^k\to \braces{0,1}$ and any hypercube $H$, we define $\mu(H):=\sum_{x\in H} \mu(x)$ as the output of the counting oracle. 
    
    \paragraph{The (random) hard instances.} 
    We consider deterministic algorithms that make at most $n^c$ queries in each round, where $c<n^{1-\delta}$ for some constant $\delta \in (0,1]$. 
    We randomly partition the $n$ variables into $r=\frac{1}{4}\parens{\frac{n}{(c+2) \log n} }^{1/3}$ equal blocks $S_1, S_2, \cdots, S_r$, each with $m=n/r=4 n^{2/3} \parens{(c+2) \log n} ^{1/3}$ variables. 
    For each block $S_i$, we choose $a_i=i\cdot 12n^{1/3}\parens{(c+2)\log n}^{2/3}$ and define the set of true strings in this block using a random linear code with constraints $m-a_i$: first we independently and uniformly choose a matrix $B_i\in \braces{0,1}^{(m-a_i)\times m}$ and a vector $v_i\in \braces{0,1}^{m-a_i}$ at random for each $j\in [m]$; and then we define the boolean function $\mu_i$ as follows:
    \begin{align*}
        \forall x\in \braces{0,1}^{S_i}, \mu_i(x) = \1[B_ix=v_i],
    \end{align*}
    where all the operations are under $\F_2$.
    Then, the true strings of the entire function are defined as those projections in each of the blocks are all true, i.e., the entire function $\mu$ is then defined as the product of all $\mu_i$:
    \begin{align*}
        \forall x\in \braces{0,1}^n, \mu(x) = \prod_{i\in [r]} \mu_i(x_{S_i}).
    \end{align*}
    For any sub-hypercube $H$, parameterized by $S$ and $y$, we define $H$ restricted to $S_i$ as $H_{S_i}:=\braces{x\in \braces{0,1}^{S_i} \given x_{S\cap S_i} = y_{S\cap S_i}}$.
    Since $S_1,\cdots, S_r$ is a partition of the $n$ variables, we have $H=H_{S_1}\times H_{S_2}\times \cdots \times H_{S_r}$ according to the definition of $H$.
    With this fact, we can obtain the following lemma.
    \begin{lemma}
    \label{lem:count-oracle-value}
        For any sub-hypercube $H$ of $\braces{0,1}^n$, we have 
        \begin{align*}
            \mu(H) = \prod_{i\in [r]} \mu_i(H_{S_i}).
        \end{align*}
    \end{lemma}
    \begin{proof}
        According to the definition of the function $\mu$ and the definition of the counting oracles,
        \begin{multline*}
            \mu(H) = \sum_{x\in H} \mu(x) 
            = \sum_{x\in H} \prod_{i\in [r]} \mu_i(x_{S_i})
            \\
            =
            \sum_{x\in H_{S_1}\times \cdots \times H_{S_r}} \prod_{i\in [r]} \mu_i(x_{S_i})
            \\
            =
            \prod_{i\in [r]} \sum_{x\in H_{S_i}} \mu_i(x)
            =
            \prod_{i\in [r]} \mu_i(H_{S_i}).  \qedhere
        \end{multline*}
    \end{proof}

    For any sub-hypercube $H$ parameterized by $S$ and $y$, we define its codimension $\codim{H}$ as $|S|$, i.e., the number of variables whose values are fixed in the sub-hypercube. 
    For any function $\mu_i$, we can show that if the codimension of a query $H_{S_i}$ is $\Omega(\log n)$ greater or less than $a_i$, the query does not give any information about the randomness of $B_i, v_i$ in the construction with high probability. 
    In addition, the proof only uses the randomness of $B_1,v_1,\cdots, B_r,v_r$.

    \begin{lemma}
        \label{lem:no-info-queries}
        For any sub-hypercube $H_{S_i}$ of $\braces{0,1}^{S_i}$, if  $\codim{H_{S_i}}=d$, for any constant $c_1>0$, we have
        \begin{itemize}
            \item if $d<a_i-c_1\log n$, $\mu_i(H_{S_i}) = 2^{a_i-d}$ with probability at least $1-n^{-c_1}$, and  
            \item if $d>a_i+c_1\log n$, $\mu_i(H_{S_i}) = 0$ with probability at least $1-n^{-c_1}$,
        \end{itemize}
        where the probability is taken over the randomness of $B_i$ and $v_i$.
    \end{lemma}
    \begin{proof}
        For any $x\in \braces{0,1}^{S_i}$ and any $B_i\in \braces{0,1}^{(m-a_i)\times m}$, $\P_{v_i\sim\braces{0,1}^{m-a_i}}{B_ix = v_i}=2^{-(m-a_i)}$.
        Therefore, we can upper bound the probability of $\mu_i(H_{S_i})\neq 0$ as follows.
        \begin{align*}
            \P_{B_i,v_i}{\mu_i(H_{S_i}) \neq 0} 
            &
            =
            \P_{B_i,v_i}{\exists x\in H_{S_i}, \mu_i(x) = 1}
            \\
            &
            \leq
            \sum_{x\in H_{S_i}} \E_{B_i,v_i}{B_ix=v_i}
            =
            2^{-(m-a_i)} \cdot \abs{H_{S_i}}
        \end{align*}
        Since $\codim{H_{S_i}} = d$, $\abs{H_{S_i}} = 2^{m-d}$. 
        We have $\P_{B_i,v_i}{\mu_i(H_{S_i}) \neq 0}\leq 2^{m-d-(m-a_i)} = 2^{a_i-d}$.
        If $d>a_i+c_1\log n$, we have $\mu_i(H_{S_i}) \neq 0$ with probability at most $n^{-c_1}$.

        On the other hand, we consider the number of solutions $x\in \braces{0,1}^{S_i}$ for $B_ix=v_i$ when $d\leq a_i$.
        For any sub-hypercube $H_{S_i}$ which is parameterized by $S\subseteq S_i$ and $y\in \braces{0,1}^S$ (i.e., $H_{S_i}=\braces{x\in \braces{0,1}^{S_i} \given x_{S}=y_{S}}$) and has codimension $d$ (i.e., $|S|=d$), we can characterize $H_{S_i}$ by $d$ linear equations: $
            \forall j\in S, e_j^Tx = y_j$,
        where $e_j$ denotes the indicator vector having value $1$ in the $j$-th entry and having value $0$ in all other entries.
        Therefore, the set $\braces{x\in\braces{0,1}^{S_i} \given x\in H_{S_i}, B_ix=v_i}$ can be characterized by $d+m-a_i$ linear equations: $
            \forall j\in S, e_j^Tx = y_j$ and $B_ix = v_i$.
        
        \begin{lemma}
        \label{lem:numsol-linear-equ}
            For any $m\leq n$, $A\in \braces{0,1}^{m\times n}$ and $b\in \braces{0,1}^m$, if $\rank(A)=m$, then there are $2^{n-m}$ solutions $x\in \braces{0,1}^n$ for the linear equation $Ax=b$ (under $\F_2$).
        \end{lemma}

        According to the above \cref{lem:numsol-linear-equ}, if vectors in $\braces{e_j\given j\in S}$ and in rows of $B_i$ are linearly independent under $\F_2$, $\mu_i(H_{S_i}) = \abs{\braces{x\in\braces{0,1}^{S_i} \given x\in H_{S_i}, B_ix=v_i}} = 2^{a_i-d}$.
        It is clear that vectors in $\braces{e_j\given j\in S}$ are linearly independent.
        Consider we start with $V=\braces{e_j\given j\in S}$ and insert rows in $B_i$ into $V$ one by one. 
        When the vectors in $V$ are linearly independent and we insert one row of $B_i$ into $V$, $V$ becomes linearly dependent only when the row is a linear combination of the vectors in $V$.
        Since there are at most $2^{|V|}$ such linear combinations under $\F_2$, the probability that $V$ remains linearly independent after inserting the row is $1-2^{|V|-m}$.
        After inserting all the $m-a_i$ rows into $V$, $V$ remains linearly independent with probability 
        \begin{align*}
            \prod_{k=d}^{d+m-a_i-1} 1-2^{k-m} \geq 1 - \sum_{k=d}^{d+m-a_i-1} 2^{k-m} \geq 1 - 2^{d-a_i}.
        \end{align*}
        Therefore, we have $\mu_i(H_{S_i})=2^{a_i-d}$ with probability at least $1-2^{d-a_i}$.
        In particular, if $d<a_i-c_1\log n$, we have $\mu_i(H_{S_i})=2^{a_i-d}$ with probability at least $1-n^{-c_1}$.
    \end{proof}

    On the other hand, the random partition $S_1,S_2,\cdots, S_r$ guarantees that any hypercube has approximately equal codimension in each block with high probability.

    \begin{lemma}
        \label{lem:balanced-rand-partition}
        For any $1\leq k<i\leq n$, any $c_2>0$, and any sub-hypercube $H$, the probability that $\codim{H_{S_i}}$ is in the range $\codim{H_{S_k\cup \cdots\cup S_r}}/\parens{r-k+1} \pm \sqrt{3c_2m\log n}$ is at least $1-2n^{-c_2}$, where the randomness is taken over the random partition of $S_k \cup S_{k+1} \cup \cdots \cup S_r$.
    \end{lemma}
    \begin{proof}
        For convenience, let $T = S_k\cup \cdots \cup S_r$ and $d'=\codim{H_{T}}$.
        Suppose the hypercube $H_{T}$ is parameterized by $S_{T}\subseteq T$ and $y_{T}\in \braces{0,1}^T$, i.e., $H_T = \braces{x\in \braces{0,1}^T\given x_T=y_T}$.
        Because $\codim{H_{T}}=|S_{T}|\leq |T|$, we have $d'\leq m(r-k+1)$.
        For each $\ell \in T$, let $Z_{\ell}$ denote the indicator whether the variable $\ell$ is in $S_i$. 
        Because $|T|=m(r-k+1)$ and $S_k,S_{k+1},\cdots,S_r$ is a uniform partition of $T$,
        for any $\ell\in T$, the probability that $Z_{\ell}=1$ is $\frac{1}{r-k+1}$.
        In addition, variables in $\braces{Z_{\ell}}_{\ell\in T}$ follow a permutation distribution and are thus negatively associated. 
        
        Let $Z=\sum_{\ell\in S_T} Z_{\ell}$ denote the number of variables in $S_i\cap S_T$.
        It is clear that $\E{Z}=\frac{d'}{r-k+1}$.
        According to the definition of codimension, we have $Z = \codim{H_{S_i}}$.
        Because of the Chernoff bound and $d'\leq m(r-k+1)$, for any $c_2>0$,
        \Tag{
        \[
            \P*{\abs*{Z-\frac{d'}{r-k+1}} > \sqrt{3c_2m\log n}} 
            \leq 
            2\exp\parens*{-\frac{3c_2m\log n}{3d'/(r-k+1)}}
            \leq 
            2\exp(-c_2\log n) = 2n^{-c_2}.
        \]
        }%
        \Tag<sigconf>{
        \begin{multline*}
            \P*{\abs*{Z-\frac{d'}{r-k+1}} > \sqrt{3c_2m\log n}} 
            \leq \\
            2\exp\parens*{-\frac{3c_2m\log n}{3d'/(r-k+1)}}
            \leq 
            2\exp(-c_2\log n) = 2n^{-c_2}.
        \end{multline*}
        }%
        Therefore, the probability of $\codim{H_{S_i}}$ being $d'/(r-k+1) \pm \sqrt{3c_2m\log n}$ is at least $1-2n^{-c_2}$.
    \end{proof}

    Putting the previous lemmas together, we obtain the following key lemma for the hardness of parallel search via counting. 
    If we only reveal the information about the first blocks $k-1$ (i.e., the partition $S_1,\cdots, S_{k-1}$ and the parameters to define the true strings $B_1,v_1,\cdots, B_{k-1},v_{k-1}$), the return value of any query is determined solely by the information about the first $k$ blocks with high probability.
    This lemma implies that without any information about the $k$-th block and its subsequent blocks, any algorithm that uses one round of queries can only learn about the information in the $k$-th block (with high probability).

    \begin{lemma}
        \label{lem:no-info-return-values}
        Fix any $k\in [r-1]$, any hypercube $H$ and any realization of $S_1,B_1,v_1, \cdots, S_{k-1}, B_{k-1}, v_{k-1}$. With probability at least $1-3n^{-(c+5/3)}$, 
        \begin{align}
            \label{eqn:no-info-query-value}
            \mu(H) = \parens*{\prod_{i\in [k]} \mu_i(H_{S_i})} \cdot 2^{-\codim{H} + \codim{H_{S_1\cup S_2\cup \cdots \cup S_k}} + \sum_{i>k} a_i }.
        \end{align}
        where the probability is taken over the random partition of $S_k\cup S_{k+1} \cup \cdots \cup S_r$ and the randomness of $B_k,v_k, B_{k+1}, v_{k+1}, \cdots, B_r, v_r$.
    \end{lemma}

    \begin{proof}
        According to \cref{lem:count-oracle-value}, it suffices to show that with probability $1-n^{-(c+5/3)}$, we have
        \begin{equation}
            \prod_{i=k}^{r} \mu_i(H_{S_i}) = \mu_k(H_{S_k}) \cdot 2^{-\codim{H} + \codim{H_{S_1\cup S_2\cup \cdots \cup S_k}} + \sum_{i>k} a_i }.
            \label{eqn:prods-of-mu}
        \end{equation}
        Let $d' = \codim{H_{S_k\cup S_{k+1}\cup \cdots \cup S_{r}}} $. Next, we prove the lemma by discussing two cases: $d'/(r-k+1)\geq \frac{a_{k-1}+a_k}{2}$ and $d'/(r-k+1) < \frac{a_{k-1}+a_k}{2}$.
        Before the discussion, recall that we define $m=4 n^{2/3} \parens{(c+2) \log n} ^{1/3}$ and $a_i=i\cdot 12n^{1/3}\parens{(c+2)\log n}^{2/3}$ for each $i\in [r]$.

        \paragraph{Case \#1: when $d'/(r-k+1)\geq \frac{a_{k}+a_{k+1}}{2}$.} According to the definition of $a_k$ and $a_{k+1}$, $d'/(r-k+1)\geq a_k+6n^{1/3}\parens{(c+2)\log n}^{2/3}$.
        Because of \cref{lem:balanced-rand-partition}, with probability at least $1-2n^{-(c+2)}$, 
        \begin{multline*}
            \codim{H_{S_k}} 
            \geq 
            \frac{d'}{r-k+1} - \sqrt{3(c+2)m\log n} 
            \geq \\
            a_k + 6n^{1/3}((c+2)\log n)^{2/3} - 2\sqrt{3} n^{1/3}((c+2)\log n)^{2/3}
            > \\
            a_k + n^{1/3}((c+2)\log n)^{2/3}
            >
            a_k + (c+2)\log n.
        \end{multline*}
        Because of \cref{lem:no-info-queries}, supposing 
            $\codim{H_{S_k}} > a_k + (c+2)\log n$, we have $\mu_k(H_{S_k})=0$ with probability at least $1-n^{-(c+2)}$. 
        Therefore, with probability at least $1-3n^{-(c+2)}$, both the LHS and RHS of \cref{eqn:prods-of-mu} equal $0$.

        \paragraph{Case \#2: when $d'/(r-k+1) < \frac{a_{k}+a_{k+1}}{2}$.}
        According to the definition of $a_k$ and $a_{k+1}$, $d'/(r-k+1)< a_{k+1}-6n^{1/3}\parens{(c+2)\log n}^{2/3}$.
        Because of \cref{lem:balanced-rand-partition} and the union bound, with probability at least $1-2n^{-(c+5/3)}$, for all $j>k$, 
        \begin{multline*}
            \codim{H_{S_j}} 
            \leq 
            \frac{d'}{r-k+1} + \sqrt{3(c+2)m\log n} 
            \leq \\
            a_k - 6n^{1/3}((c+2)\log n)^{2/3} + 2\sqrt{3} n^{1/3}((c+2)\log n)^{2/3}
            < \\
            a_j - n^{1/3}((c+2)\log n)^{2/3}
            <
            a_j - (c+2)\log n.
        \end{multline*}
        Because of \cref{lem:no-info-queries}, supposing 
            $\codim{H_{S_j}} < a_j - (c+2)\log n$ for any $j>k$, we have $\mu_j(H_{S_j})=2^{a_j-\codim{H_{S_j}}}$ for any $j>k$ with probability at least $1-n^{-(c+2)}$. 
        Therefore, with probability at least $1-3n^{-(c+5/3)}$, we have
        \begin{multline*}
            \prod_{i=k}^r \mu_i(H_{S_i}) 
            = 
            \mu_k(H_{S_k}) \cdot \prod_{j=k+1}^r 2^{a_j-\codim{H_{S_j}}}
            = \\
            \mu_k(H_{S_k}) \cdot 2^{\sum_{j>k} a_j-\sum_{j>k}\codim{H_{S_j}}}.
        \end{multline*}
        Since $S_1,S_2,\cdots,S_r$ is a partition, we have $\sum_{j>k} \codim{H_{S_j}} = \codim{H_{S_{k+1}\cup \cdots \cup S_r}} = \codim{H} - \codim{H_{S_1\cup \cdots \cup S_k}}$.
        Hence, we obtain \cref{eqn:prods-of-mu} for this case.
    \end{proof}

    Finally, we can establish the main theorem of this section.

    \begin{proof}[Proof of \cref{thm:hardness-search}]
        We show the following statement by induction: for any $i\in [r]$, given the sets $S_1,S_2,\cdots, S_{i-1}$ and information $B_1,v_1,\cdots, B_{i-1},v_{i-1}$, any deterministic algorithm can only find a solution $x$ with probability at most $3(r-i+1)n^{-5/3}$ after $r-i$ rounds.
        Note that given sets $S_1,S_2,\cdots, S_{i-1}$ and $B_1,v_1,\cdots, B_{i-1},v_{i-1}$, the remaining sets $S_i,\cdots, S_r$ form a uniform random partition of $[n]\setminus\parens{S_1\cup\cdots\cup S_{i-1}}$ and the remaining randomness $B_i,v_i,\cdots,B_r,v_r$ are uniformly at random.
        When $i=1$, the statement is equivalent to any deterministic algorithm cannot find a solution with probability $n^{-4/3}$ after $r$ rounds of queries.
        According to Yao's minimax principle, this implies that, for any randomized algorithm, there exists an instance such that the algorithm can only find a solution with probability at most $n^{-4/3}$ after $r-1$ rounds of queries.
        Next, we prove these statements to finish our proof.
        
        The base case is $i=r$. 
        With no queries, any deterministic algorithm returns a fixed $x^*$. 
        Since $v_r\in \F_2^{m-a_r}$ is uniformly random, $m=4n^{2/3}((c+2)\log n)^{1/3}$ and $a_r = 3n^{2/3}((c+2)\log n)^{1/3}$, the probability that $x^*$ is a solution is $\P{B_rx^*=v_r}<2^{a_i-m}<n^{-5/3}$.

        Suppose we have shown for $i=k+1$ (where $1\leq k\leq r-1$).
        Consider any deterministic algorithm $\ALG$.
        Given $S_1,S_2,\cdots, S_{k-1}$ and $B_1,v_1,\cdots, B_{k-1},v_{k-1}$, when $\ALG$ finds a solution $x^*$ in $r-k$ rounds, at least one of the following events occurs.
        \begin{itemize}
            \item There exists an $\ALG$'s first-round query such that its return value does not follow \cref{eqn:no-info-query-value}.
            \item Given $S_1,S_2,\cdots, S_k$ and $B_1,v_1,\cdots, B_k,v_k$, simulating the round $1$ queries of $\ALG$ by \cref{eqn:no-info-query-value} and running $\ALG$ from round $2$, we can find a solution $x^*$ in $r-k-1$ rounds.
        \end{itemize}
        In the first round of the algorithm, according to \cref{lem:no-info-return-values}, with probability at least $1-3n^{-5/3}$, the return values of all the $n^c$ queries follow \cref{eqn:no-info-query-value}.
        Therefore, the probability that the first event occurs is at most $3n^{-5/3}$.
        Also, note that \cref{eqn:no-info-query-value} is fully determined by the the information in the first $k$ blocks -- $S_1,S_2,\cdots, S_k$ in the partition and the bits $B_1,v_1,\cdots, B_k,v_k$ used in the construction of the first $k$ random linear codes.
        We can correctly simulate the first-round queries in the second event.
        Due to the induction hypothesis, the second event occurs with probability at most $3(r-k)n^{-5/3}$.
        Because of the union bound, $\ALG$ can find a solution with probability at most $3(r-k+1)n^{5/3}$ using $r-k$ rounds of queries.
    \end{proof}

%% file: algo-lb.tex
\section{Lower bound of \cref{alg:sample-on-hypergrid2}}\label{sec:tightness}
    \label{sec:algo-lb}
In this section, we show that our analysis of \cref{alg:sample-on-hypergrid2} is tight up to polylogarithmic factors. For concreteness, we assume the \textsc{UniversalCoupler} used is \textsc{MinCoupler}.
\begin{theorem}
    \label{thm:lb-for-our-algorithm}
    There exists an instance such that \cref{alg:sample-on-hypergrid2} terminates in $\Omega(\frac{n^{2/3}}{\log n})$ rounds with probability at least $0.99$.
\end{theorem}
The proof of this lower bound uses the explicit construction of the universal coupler used in the algorithm (\textsc{MinCoupler}), but we expect it to also hold with small modifications for \textsc{GumbelTrick}. We are unaware whether this lower bound holds for an arbitrary universal coupler.

\paragraph{Errors of the universal coupler.} Given a distribution $\mu$, we identify a set of randomness $r$ where \textsc{MinCoupler} produces different samples with constant probability for $\mu$ and $\upsilon$
when the second distribution $\upsilon$ is very different from $\mu$.
Here, we say $\upsilon$ is very different from $\mu$, if we sample $x$ from $\mu$, there is a constant probability of having $\upsilon(x)< (1-\delta)\mu(x)$ for some reasonably large $\delta$.
Recall the construction of \textsc{MinCoupler} -- we encode $r$ as pairs $(x_1,p_1), (x_2,p_2), \cdots \in [q]\times [0,1]$, choose the minimum index $i^*$ such that $p_{i^*}\leq \mu(x_{i^*})$ and let $x_{i^*}$ be the output of the universal coupler. 
The output $x_{i^*}$ follows the distribution $\mu$.
If we consider the restricted set of $r$ such that $p_{i^*}\geq (1-\delta)\mu(x_{i^*})$, for any distribution $\upsilon$ that is very different from $\mu$, \textsc{MinCoupler} will produce different samples for $\mu$ and $\upsilon$ with constant probability. 
We formalize the above argument as the following lemma, whose proof is deferred to \cref{proof:sure-mistakes-coupler}.

\begin{lemma}[Sure mistakes made by $\textsc{MinCoupler}$]
    \label{lem:sure-mistakes-coupler}
    Consider any distribution $\mu\in \Delta_q$ and any $\delta>0$.
    Suppose $\mathcal{R}(\mu, \delta)$ is the set of randomness $r=((x_1, p_1), (x_2,p_2), \cdots)$ used by $\textsc{MinCoupler}$ such that $p_{i^*}\geq (1-\delta) \mu(x_{i^*})$, where $i^*=\min\{i:p_{i} \leq \mu(x_{i})\}$ is the index chosen by \textsc{MinCoupler}.

    Then, for any distribution $\upsilon\in \Delta_q$, we have 
    \[
        \P *_{r\sim \mathcal{R}(\mu, \delta)}{
            \textsc{MinCoupler}(\mu, r) \neq \textsc{MinCoupler}(\upsilon, r)
        }
        \geq
        \frac{1-\upsilon_{\max}}{2} \cdot \P_{x\sim \mu}{\upsilon(x) < (1-\delta)\mu(x)},
    \]
    where $\upsilon_{\max}$ is the maximum mass $\max_{x\in [q]} \upsilon(x)$ of the distribution $\upsilon$.
\end{lemma}

\paragraph{The (random) hard instances.} 
Suppose $m=20n$. 
Consider parameters $y_1, y_2, \cdots, y_n \in \R^m$, which are randomly constructed in symmetry and will be stated later.
Let $z\sim N(0,I_m)$ be a random vector. 
For each $i\in [n]$, let the variable $X_i=\mathsf{round}(\angles{y_i,z})$, where the rounding function $\mathsf{round}(x)$ is constructed as follows.
\begin{align*}
    \mathsf{round}(x) = \min \braces*{n^4\log n, \max \braces*{-n^4\log n, \floor*{n^4x} } }.
\end{align*}
To prove the lower bound, we consider the parameters $y_1,y_2,\cdots, y_n$ as i.i.d.s following the distribution $N(0,\frac{1}{m}I_m)$ and use $V_{i}=\angles{y_i,z}$ to denote the variables before rounding.
It can be shown that, with high probability, $X_i$ equals the floor of $n^4V_i$ for any $i\in [n]$.
\begin{lemma}
\label{lem:diff-between-XandV}
    For any $y_1,y_2,\cdots, y_n$, with probability at least $1-n^{-\Omega(\log n)}$, $\forall i\in [n], X_i = \floor*{n^4V_i}$.
\end{lemma}

In addition, $\norm{y_i}^2$ follows the $\chi_m^2$ distribution (with a $\frac{1}{m}$ factor). 
We can use the following to show that $\norm{y_i}^2$ is concentrated within $1\pm \epsilon$ with probability at least $1-2^{-n^{\Omega(1)}}$ for any constant $\epsilon>0$.

\begin{lemma}[Laurent-Massart bound~{\cite[Lemma 1,][]{LM00}}]
    Let $y\sim N(0,I_m)$ and $a\in \R_{\geq 0}^m$. 
    Let $Z = \sum_{i\in [m]} a_i(y_i^2-1)$.
    Then, for any $x\geq 0$,
    \begin{align*}
        \P *{|Z|>2\norm{a}_2\sqrt{x} + 2\norm{a}_{\infty} x} \leq 2\exp(-x).
    \end{align*}
\end{lemma}

Suppose $y'_1,\cdots, y'_n$ are the vectors generated by the Gram-Schmidt orthogonalization procedure on $y_1,\cdots, y_n$. 
With this lemma and the fact that $y_1,\cdots, y_n$ are i.i.d. Gaussians, if we write each $y'_i$ as a linear expression of $y_1,\cdots, y_n$, the coefficients can be bounded polynomially in $n$.
The proof is deferred to \cref{proof:coef-of-gramschmidt}.

\begin{lemma}
    \label{lem:coef-of-gramschmidt}
    Suppose $n\leq m/20$ and $y_1,y_2,\cdots,y_n$ are i.i.d. vectors following the distribution $N(0, \frac{1}{m}I_m)$.
    Consider $y'_1, y'_2,\cdots, y'_n$ as the vectors generated by the Gram-Schmidt orthogonalization procedure on $y_1, y_2,\cdots, y_n$.
    Suppose $\forall j\in [n], y'_j = \sum_{k=1}^j c_{jk}y_k$.
    With probability at least $1-O(n^{-3})$, for any $\ell\in [n]$, 
        $\sum_{j=1}^n c_{j\ell}^2 \leq 2n$.
\end{lemma}

Next, we show that if $i-a=\Omega(n^{1/3})$, the conditional distribution of $X_{i} | \braces{X_{j}}_{j\in [a]}$ and $X_{i} | \braces{X_{j}}_{j\in [i-1]}$ can be very different under random conditioning of $X_1,\cdots, X_{i-1}$ for $\delta = \Omega(n^{-1/3})$. 
We formalize it in \cref{lem:ub-diff-between-conditionings}.
Its proof can be summarized by the following 3 key ingredients:
\begin{enumerate}
    \item $V_{i} | \braces{V_{j}}_{j\in [a]}$ and $V_{i} | \braces{V_{j}}_{j\in [i-1]}$ can be neatly expressed as (randomly constructed) normal distributions whose means have a difference $\Omega(n^{-1/3})$ with a constant probability and whose variances are both $\Omega(1)$ with a high probability.
    \item Rounding two normal distributions that satisfy the above two properties gives two very different discrete distributions for $\delta=\Omega(n^{-1/3})$, implying that $X_{i} | \braces{V_{j}}_{j\in [a]}$ and $X_{i} | \braces{V_{j}}_{j\in [i-1]}$ are very different.
    \item Using \cref{lem:coef-of-gramschmidt} to handle random noise produced by conditioning on $\{X_j\}$ instead of $\{V_j\}$.
\end{enumerate}
The full proof is deferred to \cref{proof:ub-diff-between-conditionings}.

\begin{lemma}
    \label{lem:ub-diff-between-conditionings}
    Consider any $a,i\in [n]$ such that $i-a> 40n^{1/3}$.
    Suppose $\mu, \upsilon$ are the (randomly constructed) distributions of $X_i | \braces{X_j}_{j\in [a]}$ and $X_i | \braces{X_j}_{j\in [i-1]}$.
    There exist constants $c_1>0,c_2\in (0,1)$ such that 
    \begin{align*}
        \mathbb{E} _{y_1,\cdots, y_n}
        \E *_{X_1,X_2,\cdots, X_{i-1}}{
           (1-\upsilon_{\max}) \cdot \P *_{k\sim \mu}{\frac{\upsilon(k)}{\mu(k)} \leq 1-c_1\cdot n^{-1/3} }
        } 
        \geq 
        c_2.
    \end{align*}
\end{lemma}

\paragraph{$\widetilde{\Omega}(n^{2/3})$ lower bound of \cref{alg:sample-on-hypergrid2}.} In the rest of this section, we suppose $c_1,c_2$ are the constants stated in \cref{lem:ub-diff-between-conditionings}.
We can show in \cref{lem:constant-fail-prob-instances}, the algorithm only makes small progress (i.e., $a$ increases by $O(n^{1/3})$) with probability at least $\Omega(1)$ in each round before termination.
\begin{lemma}
    \label{lem:constant-fail-prob-instances}
    Suppose $\forall i\in [n], \sigma(i)=i$.
    For any $a_0\in [n-1]\cup \braces{0}$, if we initiate \cref{alg:sample-on-hypergrid2} with $a=a_0$ and $\forall i\in[a_0], x^0_i = \tilde{x}_i(\sigma, u)$, the probability that the algorithm will have $a\leq a_0+(40+c_1^{-1}\cdot \ceil{\log(4c_2^{-1})} )n^{1/3}$ after one round is at least $\frac{c_2}{4}$, where the probability is taken over the randomness $y_1,y_2,\cdots, y_n$ used in the constructions of the instances and the randomness $u_1,\cdots, u_n$ used by the algorithm. 
\end{lemma}
\begin{proof}
    Let $obj=a_0+(40+c_1^{-1}\cdot \ceil{\log(4c_2^{-1})})n^{1/3}$ be the objective value of $a$ after one round. 
    Let $\mathcal{I}$ be the set of integers in $[a_0+40n^{1/3}, obj]$.
    For each $i>a$, let $\mu_i$ be the distribution of $X_i$ conditioning on $\braces{X_j=\tilde{x}_j(\sigma, u)}_{j\in [a]}$
    and let $\upsilon_i$ be the distribution of $X_i$ conditioning on $\braces{X_j=\tilde{x}_j(\sigma, u)}_{j\in [i-1]}$.
    Let $\mathcal{A}_i$ be the event that $\forall i<j\in \mathcal{I}, u_j\notin \mathcal{R}\parens{\mu_{j}, c_1\cdot n^{-1/3}}$ and $u_i\in \mathcal{R}\parens{\mu_{j}, c_1\cdot n^{-1/3}}$.
    It is clear that $\mathcal{A}_i$ are disjoint events.
    Note that we have $a>obj$ only if for any $i\in \mathcal{I}$,
    $
        \textsc{MinCoupler}(\mu_{i},u_{i}) = \textsc{MinCoupler}(\upsilon_{i},u_{i})
    $.
    Therefore, we have 
    \begin{align*}
        \P{a>obj} 
        &\leq 
        \P{\forall i\in \mathcal{I}, \textsc{MinCoupler}(\mu_{i},u_{i}) = \textsc{MinCoupler}(\upsilon_{i},u_{i})}
        \\
        &=
        1 - \P{\exists i\in \mathcal{I}, \textsc{MinCoupler}(\mu_{i},u_{i}) \neq \textsc{MinCoupler}(\upsilon_{i},u_{i})}
        \\
        &\leq
        1 - \P{\exists i\in \mathcal{I}, u_i\in \mathcal{R}(\mu_i, c_1\cdot n^{-1/3}) \text{ and } \textsc{MinCoupler}(\mu_{i},u_{i}) \neq \textsc{MinCoupler}(\upsilon_{i},u_{i})}
        \\
        &\leq 
        1 - \sum_{i\in \mathcal{I}} \P{\mathcal{A}_i} \cdot \P*{\textsc{MinCoupler}(\mu_{i},u_{i}) \neq \textsc{MinCoupler}(\upsilon_{i},u_{i}) \given \mathcal{A}_i}
        \\
        &\leq
        1 - \parens*{\sum_{i\in \mathcal{I}} \P{\mathcal{A}_i}} \cdot \min_{i\in \mathcal{I}} \P*{\textsc{MinCoupler}(\mu_{i},u_{i}) \neq \textsc{MinCoupler}(\upsilon_{i},u_{i}) \given \mathcal{A}_i}.
    \end{align*}

    Since $\mathcal{A}_i$ are disjoint events, $\sum_{i\in \mathcal{I}} \P{\mathcal{A}_i}$ equals the probability that there exists $i\in \mathcal{I}$ such that $u_i\in \mathcal{R}(\mu_i,c_1\cdot n^{-1/3})$. 
    Because $u_{i}, i\in \mathcal{I}$ are i.i.d.s in $[0,1]$ and $\abs{\mathcal{I}} = c_1^{-1}\ceil{\log(4c_2^{-1})}n^{1/3}$, this sum of probabilities be lower bounded as follows:
    \begin{align*}
        \sum_{i\in \mathcal{I}} \P{\mathcal{A}_i} = 
        \P*{\exists i\in \mathcal{I}, u_i\in \mathcal{R}\parens*{\mu_{i}, c_1\cdot n^{-1/3}} } \geq 1-\parens*{1-c_1\cdot n^{-1/3}}^{c_1^{-1} \log(4c_2^{-1}) n^{1/3}} \geq 1 - \frac{c_2}{4}.
    \end{align*}
    On the other hand, for any $i\in \mathcal{I}$, 
    conditioning on the event $\mathcal{A}_i$, we have $u_{i}\sim \mathcal{R}(\mu_{i}, c_1\cdot n^{-1/3})$ and $u_1,\cdots, u_{i-1}$ are uniform i.i.d.s in $[0,1]$.
    Therefore, $\tilde{x}_1(\sigma,u),\cdots, \tilde{x}_{i-1}(\sigma, u)$ follows its original marginal distribution of $X_1,\cdots, X_{i-1}$ after conditioning on $\mathcal{A}_i$. For any choice of $i\in \mathcal{I}$, we have the following lower bound:
    \begin{align*}  
        & \P*{\textsc{MinCoupler}(\mu_{i},u_{i}) \neq \textsc{MinCoupler}(\upsilon_{i},u_{i}) \given  \mathcal{A}_i}
        \\
        = & \mathbb{E}_{y_1,\cdots y_n} \E *_{u_1,\cdots, u_{i-1}}{ \P *_{u_{i}\sim \mathcal{R}(\mu_{i}, c_1\cdot n^{-1/3})}{\textsc{MinCoupler}(\mu_{i}, u_{i}) \neq \textsc{MinCoupler}(\upsilon_{i}, u_{i})} }
        \\
        \geq & \mathbb{E}_{y_1,\cdots y_n} \E *_{u_1,\cdots, u_{i-1}}{ \frac{1-(\upsilon_i)_{\max}}{2} \cdot \P *_{x\sim \mu_{i}}{\upsilon_{i}(x)\leq (1-c_1\cdot n^{-1/3})\mu_{i}(x)} }
        \geq \frac{c_2}{2}.
    \end{align*}
    Therefore, $\P{a>obj}\geq 1-(1-\frac{c_2}{4})\cdot \frac{c_2}{2} \geq 1-\frac{c_2}{4}$.
\end{proof}

However, this constant probability does not suffice to show the $\widetilde{\Omega}(n^{2/3})$ lower bound.
This is because we have not eliminated the possibility that the algorithm can terminate with a constant probability in each round. 
Next, we boost this probability of small progress to $1-n^{-\Omega(1)}$ by constructing a new instance with $\poly\log(n)$ i.i.d. such instances.
More specifically, let $g=20c_2^{-1}\log n$ be the number of groups. 
The parameters of a new instance are the vectors $y_1,y_2,\cdots, y_n\in \R^m$ and the group numbers $h(1),h(2),\cdots, h(n)\in [g]$.
Let $z_1,z_2,\cdots, z_g\sim N(0,I_m)$ be i.i.d. random Gaussian vectors.
Then, the variables $X_1, \cdots, X_n$ in the new instance are defined as follows:
\begin{align*}
    X_i = \mathsf{round}\parens*{
        \angles{y_i, z_{h(i)}}
    }.
\end{align*}    
With this new construction, we can show that the algorithm will have small progress in each round with high probability. 

\begin{lemma}
    \label{lem:high-fail-prob-instances}
    Suppose $\forall i\in [n], \sigma(i)=i$.
    For any $a_0\in [n-1]\cup \braces{0}$, if we initiate \cref{alg:sample-on-hypergrid2} with $a=a_0$ and $\forall i\in[a_0], x^0_i = \tilde{x}_i(\sigma, u)$, the probability that the algorithm will have $a\geq a_0+2(40+c_1^{-1}\cdot \ceil{\log(4c_2^{-1})} )n^{1/3}g$ after one round is at most $(1+o(1))n^{-5}$, where the probability is taken over the randomness $y_1,y_2,\cdots, y_n, h(1), \cdots, h(n)$ used in constructions of the instances and the randomness $u_1,\cdots, u_n$ used by the algorithm. 
\end{lemma}
\begin{proof}
    For convenience, we use $obj=a_0+2(40+c_1^{-1}\cdot \ceil{\log(4c_2^{-1})})n^{1/3}g$ to denote the objective position for the algorithm to reach after one round.
    Let $\mathcal{I}$ be the set of integers in $[a_0+1, obj]$.
    For any $h\in [g]$, let $\mathcal{I}'_h$ be the set of $i\in [obj]$ such that $h(i)=h$.
    With probability $1-2^{-n^{\Omega(1)}}$ over the choices of $\braces{h(i)\given i\in \mathcal{I}}$, for any $h\in [g]$, the size of $\braces{i\in \mathcal{I}\given h(i)=h}$ (i.e., $\mathcal{I}\cap \mathcal{I}'_h$) is at least $(40+c_1^{-1}\cdot \ceil{\log(4c_2^{-1})})n^{1/3}$.
    If $a\geq obj$ after one round of \cref{alg:sample-on-hypergrid2}, there should be
    \begin{align}
        \label{eqn:cond-large-progress}
        \forall i\in \mathcal{I}, \quad \textsc{MinCoupler}\parens*{\parens*{X_i \given \braces*{X_j}_{j\in [a_0]}}, u_i} = \textsc{MinCoupler}\parens*{\parens*{X_i \given \braces*{X_j}_{j\in [i-1]}}, u_i}
    \end{align}
    Note that for any $i,j\in [obj]$ such that $h(i)\neq h(j)$, $X_i, X_j$ are independent.
    \cref{eqn:cond-large-progress} is equivalent to $\forall h\in [g]$,
    \begin{align*}
        \forall i\in \mathcal{I}\cap \mathcal{I}'_h, ~~ \textsc{MinCoupler}\parens*{\parens*{X_i \given \braces*{X_j}_{j\in [a_0]\cap \mathcal{I}'_{h}}}, u_i} = \textsc{MinCoupler}\parens*{\parens*{X_i \given \braces*{X_j}_{j\in [i-1]\cap \mathcal{I}'_{h}}}, u_i}
    \end{align*}
    Because of \cref{lem:constant-fail-prob-instances}, for each $h\in [g]$, it happens with probability at most $1-\frac{c_2}{4}$.
    Since $u_1,\cdots, u_n$ are i.i.d.s, under these choices of $\braces{h(i)\given i\in \mathcal{I}}$, the probability that $a\geq obj$ after one round is at most 
    \begin{align*}
        \parens*{1-\frac{c_2}{4}}^{g} = n^{-5}.
    \end{align*}
    According to the union bound, we complete the proof.
\end{proof}

Finally, we establish the main theorem of this section. 

\begin{proof}[proof of \cref{thm:lb-for-our-algorithm}]
Let $R(X,\sigma,u)$ be the number of rounds of \cref{alg:sample-on-hypergrid2} on variables $X$, using randomness $\sigma, u$. 
It suffices to show that 
\begin{align}
\label{eqn:lb-for-our-algorithm}
\P *_{\sigma, u, y, h}{R(X,\sigma, u)\geq \frac{n^{2/3}}{40c_2^{-1}(40+c_1^{-1}\cdot \ceil{\log(4c_2^{-1})})\log n}} \geq 0.99.
\end{align}
Since $y_1,h(1),\cdots,y_n, h(n)$ are constructed in symmetry, for any permutation $\sigma \in \mathcal{S}_n$, $X_{\sigma(1)},\cdots, X_{\sigma(n)}$ are identically distributed as $X_1,\cdots, X_n$.
Therefore, it suffices to show \cref{eqn:lb-for-our-algorithm} assuming $\sigma(i)=i$ for any $i\in [n]$.
According to \cref{lem:high-fail-prob-instances} and the union bound, with probability at least $1-n^{-4}$ over the choice of $y,h,u$, for any initialization of $a$, we can increase $a$ by at most $40c_2^{-1}(40+c_1^{-1}\cdot \ceil{\log(4c_2^{-1})})n^{1/3}\log n$.
In this case, the round complexity of \cref{alg:sample-on-hypergrid2} is at least $\frac{n^{2/3}}{40c_2^{-1}(40+c_1^{-1}\cdot \ceil{\log(4c_2^{-1})})\log n}$, and thus we obtain~\cref{eqn:lb-for-our-algorithm}.
\end{proof}

%% file: acmacks.tex
\begin{acks}
	\input{acks}	
\end{acks}

%% file: app-proof.tex
\section{Deferred Proofs}
\subsection{Proof of \cref{lem:approx-counting-equiv-under-good}}
\label{proof:approx-counting-equiv-under-good}
We use $x^t,y^t,a^t$ to denote the intermediate variables used in the algorithm with the exact counting oracle, and use $\hat{x}^t,\hat{y}^t,\hat{a}^t$ to denote the intermediate variables used in the algorithm with the approximate counting oracle.
    We prove the claim that $a^t=\hat{a}^t$ and $x^t_{\sigma(j)}=\hat{x}^t_{\sigma(j)}=\tilde{x}_j(\sigma,u)$ for any $j\in [a^t]$ after each round $t$ we compute $a^t$.
    
    When $t=0$, $a^t=\hat{a}^t=0$ and this claim is clearly true.
    Consider $t\geq 1$. 
    Suppose we have proved the claim for $t-1$.
    Because of the induction hypothesis and the fact that $(\sigma, u)$ is good, $y^t=\hat{y}^t$. 
    Because of the definition of $a^t$ and \cref{lem:fixed-after-a}, we have $y^t_{\sigma(j)}=x^t_{\sigma(j)}=\tilde{x}_j(\sigma, u)$ for any $j\in [a^t-1]$ and we have $y^t_{\sigma(a^t)}\neq x^t_{\sigma(a^t)}=\tilde{x}_{a^t}(\sigma, u)$. 
    Because of $\hat{y}^t=y^t$ and the induction hypothesis, we have $\hat{x}^t_{\sigma(j)} = x^{t}_{\sigma(j)}=\tilde{x}_j(\sigma, u)$ for any $j\in [a^t]$. 
    Therefore, according to the definition of $\hat{a}^t$, $\hat{a}^t=a^t$.

    Finally, we show how to obtain the lemma using this claim. 
    If the algorithm with the exact counting oracle terminates with $a^t=n$ in some round, the claim directly implies that the algorithm with the approximate counting oracle terminates in the same round and outputs the vector.
    Otherwise, suppose $t$ is the round in which the algorithm with the exact counting oracle terminates. 
    There is $y^t=x^t$.
    As discussed above, the claim gives $\hat{y}^t=y^t$.
    Therefore, $\hat{y}^t_{\sigma(i)}=x^t_{\sigma(i)}=\tilde{x}_i(\sigma, u)$ for any $i\in [n]$. 
    Since $(\sigma, u)$ is good, the algorithm with the approximate counting oracle will generate $\hat{x}^t = \hat{y}^t$ and terminate in this round.
    The output $\hat{x}^t$ is thus the same as $x^t$.

\subsection{Proof of \cref{lem:approx-counting-good-prob}}
\label{proof:approx-counting-good-prob}

    Note that we use $q$ queries of $\hat{\mu}$ to compute each $\upsilon_{i|a}(\sigma, u)$. 
    With probability at least $1-n^2q\delta$, all queries $\hat{\mu}$ we use while computing $\upsilon_{i|a}(\sigma, u)$s satisfy \cref{eqn:approx-counting}.
    Under these circumstances, for any $0\leq a< i\leq n$ and any $x\in [q]$,
    \begin{align*}
        \parens*{\upsilon_{i|a}(\sigma, u)}(x) \geq \frac{1-\epsilon}{1+\epsilon} \cdot \P *{X_{\sigma(i)}=x \given \braces{X_{\sigma(j)}=\tilde{x}_j(\sigma,u)}_{j\in [a]}}.
    \end{align*}
    Therefore, $\tv *{X_{\sigma(i)}\given \braces{X_{\sigma(j)}=\tilde{x}_j(\sigma,u)}_{j\in [a]}, \upsilon_{i|a}(\sigma, u)} \leq 1-\frac{1-\epsilon}{1+\epsilon} \leq 2\epsilon$.
    Recall that the universal coupler guarantees $\P{\textsc{UniversalCoupler}(\mu, u)\neq \textsc{UniversalCoupler}(\upsilon, u)}\leq 2\tv{\mu, \upsilon}$ for any two distributions $\mu, \upsilon$. 
    Under these circumstances, due to the union bound, $(\sigma,u)$ is good with probability at least $1-4n^2\epsilon$. 
    Putting things together, any $(\sigma,u)$ is good with probability at least $(1-n^2q\delta)\cdot (1-4n^2\epsilon) \geq 1-O(n^2\epsilon+n^2q\delta)$.

\subsection{Proof of \cref{lem:sure-mistakes-coupler}}
\label{proof:sure-mistakes-coupler}
Let $i^*_{\mu} = \min \braces{i \given p_i\leq \mu(x_i)}$ and $i^*_{\upsilon} = \min \braces{i\given p_i\leq \upsilon(x_i)}$.
    According to the definition of {\sc{MinCoupler}}, $\textsc{MinCoupler}(\mu, r) \neq \textsc{MinCoupler}(\upsilon, r)$ if and only if 
    $x_{i^*_{\mu}} \neq x_{i^*_{\upsilon}}$.
    Observe that
    \begin{align*}
        \P *_{r\sim \mathcal{R}(\mu,\delta)}{x_{i^*_{\mu}} \neq x_{i^*_{\upsilon}}} 
        &=
        \frac{
            \P *_{r\sim [0,1]}{i^*_{\mu} < i^*_{\upsilon},~ r\in \mathcal{R}(\mu,\delta)} 
            \cdot 
            \P *{x_{i_{\mu}^*}\neq x_{i_{\upsilon}^*} \given i^*_{\mu} < i^*_{\upsilon},~r\in \mathcal{R}(\mu,\delta)}
        }
        {\P*_{r\sim [0,1]}{r\in \mathcal{R}(\mu, \delta)}}.
    \end{align*}
    On the denominator, we have,
    \begin{align*}
        \P*_{r\sim [0,1]}{r\in \mathcal{R}(\mu, \delta)} 
        &=
        \sum_{i\geq 1}
        \P*_{r\sim [0,1]}{i^*_{\mu} = i} \cdot \P *_{r\sim [0,1]}{p_i \geq (1-\delta) \mu(x_i) \given p_i \leq \mu(x_i)}
        \\
        &=
        \sum_{i\geq 1}
        \P*_{r\sim [0,1]}{i^*_{\mu} = i} \cdot \delta = \delta. 
        \tag{$p_i\sim [0,1]$}
    \end{align*}
    On the other hand, let $i^*_{\min}=\min\braces{i^*_{\mu},i^*_{\upsilon}}$. 
    As discussed in the preliminary, we have $i^*_{\min} = \min\braces{i\given p_i \leq \max\braces{\mu(x_i), \upsilon(x_i)}}$.
    For any $i\geq 1$,
    \begin{align*}
        &\P *_{r\sim [0,1]}{i^*_{\mu}<i^*_{\upsilon},~ r\in \mathcal{R}(\mu,\delta) \given i^*_{\min}=i}
        \\
        =&
        \P *_{r\sim [0,1]}{p_i>\upsilon(x_i), ~ p_i \in \bracks*{(1-\delta)\mu(x_i), ~ \mu(x_i)} \given i^*_{\min}=i}
        \\
        \geq&
        \P *_{r\sim [0,1]}{\upsilon(x_i)<(1-\delta)\mu(x_i),~ p_i \in \bracks*{(1-\delta)\mu(x_i), ~ \mu(x_i)} \given i^*_{\min}=i}
        \\
        =&
        \P *_{r\sim [0,1]}{\upsilon(x_i)<(1-\delta)\mu(x_i), ~ p_i \in \bracks*{(1-\delta)\mu(x_i), ~ \mu(x_i)} \given p_i \leq \max\braces{\mu(x_i),\upsilon(x_i)}}
        \tag{$(x_i,p_i)$s are drawn independently}
        \\
        =&
        \frac{\P *_{r\sim [0,1]}{\upsilon(x_i)<(1-\delta)\mu(x_i),~ p_i \in \bracks*{(1-\delta)\mu(x_i), ~ \mu(x_i)} }}
        {\P*{p_i \leq \max\braces{\mu(x_i),\upsilon(x_i)}}}
        \\
        =&
        \frac{q^{-1}\delta \sum_{x\in [q]} \mu(x)}
        {q^{-1} \sum_{x\in [q]} \max\braces{\mu(x),\upsilon(x)}}
        \cdot \P *_{r\sim [0,1]}{\upsilon(x_i)<(1-\delta)\mu(x_i) \given p_i \in \bracks*{(1-\delta)\mu(x_i), ~ \mu(x_i)} }
        \tag{$x_i\sim [q], p_i\sim [0,1]$}
        \\
        =&
        \frac{ \delta }{ 1+\tv{\mu, \upsilon}} \cdot \P *_{x\sim \mu}{\upsilon(x)<(1-\delta)\mu(x) }
        \geq
        \frac{ \delta }{ 2} \cdot \P *_{x\sim \mu}{\upsilon(x)<(1-\delta)\mu(x) }.
    \end{align*}
    Therefore, we have 
    \[
        \P_{r\sim [0,1]}{i^*_{\mu}<i^*_{\upsilon}, r\in \mathcal{R}(\mu,\delta)}\geq \frac{\delta}{2}\cdot \P_{x\sim \mu}{\upsilon(x)<(1-\delta)\mu(x)}.
    \]
    In addition, for any $i\geq 1$, the event that $i=i^*_{\mu}<i^*_{\upsilon}$ and $r\in \mathcal{R}(\mu,\delta)$ is independent of the value of any $x_j$ for $j>i$.
    Since for any $j\geq 1$ and any (possibly random) $x\in [q]$, $\P{x_j\neq x}\geq 1-\upsilon_{\max}$, 
    we have 
    \[
    \P{x_{i^*_{\mu}}\neq x_{i^*_{\upsilon}} \given i^*_{\mu}<i^*_{\upsilon}, r\in \mathcal{R}(\mu,\delta)} \geq 1-\upsilon_{\max}.
    \]
    Therefore, on the numerator, we have,
    \begin{align*}
        \P *_{r\sim \mathcal{R}(\mu,\delta)}{x_{i^*_{\mu}} \neq x_{i^*_{\upsilon}}}
        &\geq 
        \frac{ \delta (1-\upsilon_{\max}) }{ 2} \cdot \P *_{x\sim \mu}{\upsilon(x)<(1-\delta)\mu(x) }.
    \end{align*}
    Putting things together, we obtain that $\P_{r\sim \mathcal{R}(\mu,\delta)}{x_{i^*_{\mu}}\neq x_{i^*_{\upsilon}}}\geq \frac{1-\upsilon_{\max}}{2}\cdot \P *_{x\sim \mu}{\upsilon(x)<(1-\delta)\mu(x) }$ and thus the lemma.

\subsection{Proof of \cref{lem:coef-of-gramschmidt}}
\label{proof:coef-of-gramschmidt}
    We shall prove this lemma by induction on $i$ of the following statements: 
    with probability at least $1-i\cdot O(n^{-4})$, we have
    \begin{align}
    \label{eqn:coef-of-gramschmidt}
    \forall \ell\in [i], \quad \sum_{j=1}^i c_{j\ell}^2 \leq 2 \cdot \parens*{1+\frac{12\log n}{m}}^i
    \end{align}
    When $i=1$, Gram-Schmidt procedure gives $y'_1=y_1/\norm{y_1}_2$.
    Because of the Laurent-Massart bound, $\norm{y_1}_2^2\in 1\pm 0.1$ with probability at least $1-2^{-n^{\Omega(1)}}$. 
    Therefore, $c_{11}^2 \leq 1.2$ with probability $1-2^{-n^{\Omega(1)}}$.

    Suppose $k>1$ and we have shown the statements for any $i<k$.
    Next, we show the statements for $i=k$.
    We only consider the randomness of $y_k$ and consider the cases where \cref{eqn:coef-of-gramschmidt} holds and $c_{jj}^2 \leq 1.2$ holds for any $j\in [i]$.
    According to the Gram-Schmidt procedure, $y'_k$ is obtained as follows: $y''_k = y_k - \sum_{j\in [k-1]} \angles{y_k, y'_j} y'_j$ and $y'_k = y''_k/\norm{y''_k}_2$.
    Since $y_k\sim N(0, \frac{1}{m}I_m)$ and $\braces{y'_j}_{j\in [k-1]}$ are orthonormal, $\angles{y_k,y'_j}$ are i.i.d.s following $N(0,\frac{1}{m})$. 
    Suppose $y''_k = y_k + \sum_{j\in [k-1]} c'_{kj}y_j$.
    Because we have 
    \begin{align*}
        \sum_{j=1}^{k-1} \angles{y_k,y'_j} y'_j 
        = 
        \sum_{j=1}^{k-1} \angles{y_k,y_j'} \cdot 
        \sum_{\ell=1}^j c_{j\ell} y_{\ell}
        =
        \sum_{\ell=1}^{k-1} 
        \sum_{j=\ell}^{k-1} c_{j\ell}\angles{y_k, y'_j}
        y_{\ell},
    \end{align*}
    we have $c'_{kk}=1$ and $c'_{k\ell}\sim N\big(0, \frac{1}{m}\sum_{j=\ell}^{k-1} c_{j\ell}^2\big)$ for any $\ell<k$.
    Therefore, for any $\ell\in [k-1]$, with probability at least $1-O(n^{-5})$,
    \begin{align*}
        c'^{2}_{k\ell} 
        \leq 
        \frac{10\log n}{m} \cdot 
        \sum_{j=\ell}^{k-1} c_{j\ell}^2 
        \leq 
        \frac{10\log n}{m} 
        \cdot 2\parens*{1+\frac{12\log n}{m}}^{k-1} 
    \end{align*}
    Then, according to the second step of obtaining $y'_k$, we have $c_{k\ell}^2=c'^2_{k\ell}/\norm{y''_k}_2^2$ and $c_{kk}^2=1/\norm{y''_k}_2^2$.
    Note that $\norm{y''_k}_2^2 = \norm{y_k}_2^2 - \sum_{j\in [k-1]} (\angles{y_k,y'_j})^2$.
    Since $\angles{y_k,y'_j}\sim N(0, \frac{1}{m})$ and $k<m/20$, we have $\norm{y''_k}_2^2 \geq 0.9$ with probability at least $1-2^{-n^{\Omega(1)}}$.
    Therefore, with probability at least $1-2^{-n^{\Omega(1)}}$, $c_{kk}^2\leq 2$, and with probability at least $1-O(n^{-4})$, for any $\ell\in [k-1]$, we have
    $$c_{k\ell}^2 \leq c'^2_{k\ell}/0.9 \leq \frac{12\log n}{m} \cdot 2 \parens*{1+\frac{12\log n}{m}}^{k-1}.$$
    Since \cref{eqn:coef-of-gramschmidt} for $i=k-1$ holds with probability at least $1-(k-1)O(n^{-4})$, we obtain the proof for $i=k$.

    Finally, because $n\leq m/20$, for any $\ell\in [n]$, we have
    \begin{align*}
        \sum_{j\in [n]} c_{j\ell}^2 \leq 2\cdot \parens*{1+\frac{12\log n}{m}}^n \leq 2\cdot \parens*{1+\frac{12\log n}{m}}^{m/20}\leq 2n.
    \end{align*}

\subsection{Proof of \cref{lem:ub-diff-between-conditionings}}
\label{proof:ub-diff-between-conditionings}
Suppose $y'_1,y'_2,\cdots, y'_n$ are the vectors generated by the Gram-Schmidt orthogonalization procedure on $y_1,y_2,\cdots, y_n$.
    We assume that the conditions in \cref{lem:diff-between-XandV} and \cref{lem:coef-of-gramschmidt} hold, which happens with probability $1-O(n^{-3})$.
    
    Consider we expand these vectors to an orthonormal basis $y'_1,y'_2,\cdots, y'_n, y'_{n+1}, \cdots, y'_m$ of $\R^m$.
    For any $i\in [n]$ and $0\leq j<i$, suppose $y_{i|j}$ is the projection of $y_i$ on the linear span of $y_1,\cdots, y_j$ and let $y_{i|j}^\perp=y_i-y_{i|j}$.
    We have $y_{i|j} = \sum_{k=1}^j \angles{y_i, y'_j} y'_j$ and $y_{i|j}^{\perp} = \sum_{k=j+1}^m \angles{y_i,y'_j} y'_j$.
    We can rewrite $V_i = \angles{y_{i|j}, z} + \angles{y_{i|j}^{\perp}, z}$ for any $j<i$.
    Therefore, for any $a<i$, we can rewrite
    \begin{align*}
        \parens*{V_i\given \braces{X_j}_{j\in[i-1]}} &= \parens*{\angles{y_{i|a}, z} \given \braces{X_j}_{j\in [a]}} + \parens*{\angles{y_{i|i-1}-y_{i|a},z} \given \braces{X_j}_{j\in[i-1]}} + \angles{y_{i|i-1}^{\perp},z},
        \\
        \parens*{V_i\given \braces{X_j}_{j\in[a]}} &= \parens*{\angles{y_{i|a}, z} \given \braces{X_j}_{j\in[a]}} + \angles{y_{i|a}^{\perp},z}.
    \end{align*}
    Note that $\norm{y_{i|i-1}-y_{i|a}}^2 = \sum_{k=a+1}^{i-1} (\angles{y_i,y'_j})^2$. 
    Because $y_i\sim N(0,\frac{1}{m}I_m)$, $y'_j$ are orthonormal, and $y'_j (j<i)$s are independent with $y_i$, $\angles{y_i,y'_j}$ are i.i.d. variables following $N(0,\frac{1}{m})$. 
    Therefore, according to the Laurent-Massart bound, $\norm{y_{i|i-1}-y_{i|a}}^2\in (1\pm 0.01)^2 \cdot \frac{i-a-1}{m}$ with probability at least $1-2^{-n^{\Omega(1)}}$. 
    This implies $\norm{y_{i|i-1}-y_{i|a}}\geq 1.5n^{-1/3}$ with probability at least $1-2^{-n^{\Omega(1)}}$.
    Since $n\leq m/20$, we similarly have $\norm{y_{i|i-1}}^2,\norm{y_{i|a}}^2\leq 0.1$ and $\norm{y_{i|i-1}^{\perp}}^2,\norm{y_{i|a}^{\perp}}^2\in [0.9,1.1]$ with probability at least $1-2^{-n^{\Omega(1)}}$.
    We shall assume these conditions in the rest of the proof.
    
    Because of our assumption (where the condition in \cref{lem:diff-between-XandV} holds), $\forall i\in [n], V_i-X_i\in [0,n^{-4}]$.
    Observe that
    \begin{align*}
        \angles{y_{i|a},z} 
        &= 
        \sum_{j\in [a]} \angles{y_i,y'_j}\angles{y'_j, z} 
        \\
        &= 
        \sum_{j\in [a]} \angles{y_i,y'_j} \sum_{k\in [j]} c_{jk} \angles{y_k,z}
        \\
        &=
        \sum_{j\in [a]} \angles{y_i,y'_j} \sum_{k\in [j]} c_{jk} \parens*{X_k+(V_k-X_k)}
        \tag{$V_k=\angles{y_k,z}$}
        \\
        &=
        \sum_{k\in [a]} \parens*{\sum_{j=k}^{a} c_{jk} \angles{y_i,y'_j}} \cdot \parens*{X_k+(V_k-X_k)}
        \\
        &\stackrel{\text{def}}{=}
        \mathcal{E}_{1} + 
        \underbrace{\sum_{k\in [a]} \parens*{\sum_{j=k}^{a} c_{jk} \angles{y_i,y'_j}} \cdot X_k}_{B}
    \end{align*}
    According to the Cauchy-Schwarz inequality and \cref{lem:coef-of-gramschmidt}, all the coefficients $(\sum_{j=k}^{a} c_{jk} \angles{y_i,y'_j})^2 \leq (\sum_{j=1}^n c_{jk}^2)(\sum_{j=1}^a (\angles{y_i,y'_j})^2)\leq 2n^{2}\norm{y_{i|a}}^2$.
    According to our assumption of $\norm{y_{i|a}}^2\leq 0.1$,
    $\mathcal{E}_{1} \in \pm 2n^{-2}$.

    On the other hand, we similarly have 
    \begin{align*}
        \angles{y_{i|i-1} - y_{i|a},z} 
        &=
        \sum_{k\in [i-1]} \parens*{\sum_{j=\max\braces{k,a+1}}^{i-1} c_{jk} \angles{y_i,y'_j}} \cdot \parens*{X_k+(V_k-X_k)}
        \\
        &\stackrel{\text{def}}{=}
        \mathcal{E}_{2} +
        \underbrace{\sum_{k\in [i-1]} \parens*{\sum_{j=\max\braces{k,a+1}}^{i-1} c_{jk} \angles{y_i,y'_j}} \cdot X_k}_{C},
    \end{align*}
    and $\mathcal{E}_2 \in \pm 2n^{-2}$.
    Because $z\sim N(0,I_m)$, $\angles{y_{i|i-1}-y_{i|a}, z}\sim N(0, \norm{y_{i|i-1}-y_{i|a}}^2)$, and we have $\P{\angles{y_{i|i-1}-y_{i|a},z} \leq -\norm{y_{i|i-1}-y_{i|a}}}\geq 0.15$.
    Since $\mathcal{E}_2 \in \pm 2n^{-2}$ and $\norm{y_{i|i-1}-y_{i|a}}\leq -1.5n^{-1/3}$, we have $\P{C\leq -n^{-1/3}}\geq 0.15$.

    Note that conditioning on $\braces{X_j}_{j\in [i-1]}$, $B$ and $C$ are fixed. 
    Consider the cases where $C\leq -n^{-1/3}$.
    Since $z\sim N(0,I_m)$, for any valid choice of $X_1,\cdots, X_{i-1}$, the distributions $\mu \sim \parens{X_i\given \braces{X_j}_{j\in [a]}}$ and $\upsilon \sim \parens{X_i\given \braces{X_j}_{j\in [a]}}$ can be concluded as follows:
    \begin{enumerate}
        \item $\mu$ is the value applying $\mathsf{round}$ on the sum of the fixed value $B$,
        random Gaussian $N(0, \norm{y_{i|a}^{\perp}}^2)$ and random variables $\parens{\mathcal{E}_1\given \braces{X_j}_{j\in [a]}}\in \pm 2n^{-2}$.
        \item $\upsilon$ is the value applying $\mathsf{round}$ on the sum of the fixed values $B,C$, 
        random Gaussian $N(0, \norm{y_{i|i-1}^{\perp}}^2)$ and random variables $\parens{\mathcal{E}_1\given \braces{X_j}_{j\in [a]}},\parens{\mathcal{E}_2\given \braces{X_j}_{j\in [i-1]}}\in \pm 2n^{-2}$.
    \end{enumerate}
    Suppose $f_1,f_2$ are the density functions of $\parens{\mathcal{E}_1\given \braces{X_j}_{j\in [a]}}$ and $\parens{\mathcal{E}_2\given \braces{X_j}_{j\in [i-1]}}$.
    We have for any integer $k\in [-n^4\log n, n^4\log n]$,
    \begin{align*}
        \mu(k) &= \int_{x=kn^{-4}}^{(k+1)n^{-4}} \int_{e_1} f_1(e_1)\cdot \frac{\exp\parens{-(x-e_1-B)^2/2\norm{y_{i|a}^{\perp}}^2}}{\sqrt{2\pi \norm{y_{i|a}^{\perp}}^2}} de_1 dx
        \\
        \upsilon(k) &= \int_{x=kn^{-4}}^{(k+1)n^{-4}} \int_{e_1, e_2} f_1(e_1)f_2(e_2) \cdot \frac{\exp\parens{-(x-e_1-e_2-B-C)^2/2\norm{y_{i|i-1}^{\perp}}^2}}{\sqrt{2\pi \norm{y_{i|i-1}^{\perp}}^2}} de_1 de_2 dx
    \end{align*}
    If $C<0$, for any integer $k\in \braces{k\in \Z \given kn^{-4}\in (4n^{-2}+B+0.1, 4n^{-2}+B+0.2)}:=\mathcal{K}$, 
    \begin{align*}
        \mu(k) &\geq \frac{\exp\parens{-(kn^{-4}-B+7n^{-2})^2/2\norm{y_{i|a}^{\perp}}^2}}{\sqrt{2\pi} \norm{y_{i|a}^{\perp}}} \cdot n^{-4}
        \\
        \upsilon(k) &\leq \frac{\exp\parens{-(kn^{-4}-B+|C|)^2/2\norm{y_{i|i-1}^{\perp}}^2}}{\sqrt{2\pi} \norm{y_{i|i-1}^{\perp}}} \cdot n^{-4}
    \end{align*}
    In this case, $\upsilon(k)/\mu(k)$ is at most
    \begin{align*}
        &
        \frac{\norm{y_{i|i-1}^{\perp}}}{\norm{y_{i|a}^{\perp}}} \cdot \exp\parens*{
            \frac{1}{2}
            \parens*{\frac{kn^{-4}-B+7n^{-2}}{\norm{y_{i|a}^{\perp}}} - \frac{kn^{-4}-B+|C|}{\norm{y_{i|i-1}^{\perp}}}}
            \parens*{\frac{kn^{-4}-B+7n^{-2}}{\norm{y_{i|a}^{\perp}}} + \frac{kn^{-4}-B+|C|}{\norm{y_{i|i-1}^{\perp}}}}
        }
        \\
        = &
        \frac{\norm{y_{i|i-1}^{\perp}}}{\norm{y_{i|a}^{\perp}}} \cdot \exp\parens*{
            \Theta(1)
            \cdot 
            \parens*{\frac{kn^{-4}-B+7n^{-2}}{\norm{y_{i|a}^{\perp}}} - \frac{kn^{-4}-B+7n^{-2}}{\norm{y_{i|i-1}^{\perp}}} - \frac{|C|-7n^{-2}}{\norm{y_{i|i-1}^{\perp}}}}
        }
        \\
        \leq &
        \exp\parens*{
            \Theta(-|C|+7n^{-2})
        }
        =
        1-\Omega(n^{-1/3}) \tag{$\norm{y_{i|a}^{\perp}}\geq \norm{y_{i|i-1}^{\perp}}$}
    \end{align*}
    In addition, we have 
    \begin{align*}
        \sum_{k\in \mathcal{K}} \mu(k) \geq \int_{0.1+7n^{-2}}^{0.2-7n^{-2}} \frac{\exp\parens*{-x^2/2\norm{y_{i|a}^{\perp}}}}{\sqrt{2\pi}\norm{y_{i|a}^{\perp}}} dx = \Omega(1).
    \end{align*}
    
    Therefore, we conclude that $\P *_{v\sim \mu}{\frac{\upsilon(k)}{\mu(k)}\leq 1-\Omega(n^{-1/3})}\geq \Omega(1)$ if $C\leq -n^{-1/3}$.
    In addition, it is clear that $\upsilon_{\max} = O(n^{-4})$. 
    Because $C\leq -n^{-1/3}$ happens with a constant probability, we complete the proof.